\documentclass[a4paper,USenglish,cleveref,numberwithinsect,thm-restate]{lipics-v2021}
\usepackage{xurl}

\hypersetup{
    colorlinks,
    linkcolor={red!50!black},
    citecolor={blue!50!black},
    urlcolor={blue!30!black}
}

\makeatletter
\g@addto@macro\bfseries{\boldmath}
\makeatother

\PassOptionsToPackage{hyphens}{url}\usepackage{hyperref}

\usepackage[ruled]{algorithm2e}
\usepackage{colonequals}
\usepackage{stmaryrd}
\usepackage{tabularx}
\usepackage{booktabs}
\usepackage{bm}
\usepackage{tikz}
\usepackage{mathabx} 
\usetikzlibrary{decorations.pathreplacing,calligraphy} 

\usepackage{relsize}
\usepackage{stackengine}
\stackMath
\usepackage{nameref}

\makeatletter
\let\orgdescriptionlabel\descriptionlabel
\renewcommand*{\descriptionlabel}[1]{%
  \let\orglabel\label
  \let\label\@gobble
  \phantomsection
  \edef\@currentlabel{#1}%
  \let\label\orglabel
  \orgdescriptionlabel{#1}%
}
\makeatother

\newtheorem{result}{Result}
\newtheorem{fact}[theorem]{Fact}

\newcommand{\prm}{\Pr_\text{matching}}
\newcommand{\pre}{\Pr_\text{edgeCover}}
\newcommand{\pp}{p}
\newcommand{\pq}{q}
\newcommand{\pr}{r}
\newcommand{\ps}{s}
\newcommand{\np}{\bar{p}}
\newcommand{\nq}{\bar{q}}
\newcommand{\nr}{\bar{r}}
\newcommand{\ns}{\bar{s}}

\newcommand\mA{\bm{A}}
\newcommand\mB{\bm{B}}
\newcommand\mC{\bm{C}}
\newcommand\mV{\bm{V}}
\newcommand\mM{\bm{U}}
\newcommand\mMM{\bm{M}}
\newcommand\mVP{\bm{V'}}
\newcommand\mG{\bm{\Gamma}}
\newcommand\mS{\bm{S}}

\newcommand\mJ{\bm{J}}
\newcommand\calF{\mathcal{F}}
\newcommand\calG{\mathcal{G}}
\DeclareMathOperator\sub{\textrm{Sub}}
\DeclareMathOperator\probmatch{\textrm{PrMatching}}
\DeclareMathOperator\probedgecover{\textrm{PrEdgeCover}}
\DeclareMathOperator\shpmatch{\textrm{\#Matching}}
\DeclareMathOperator{\tw}{tw}

\newcommand{\shp}{\text{$\#$}\text{\rm P}}

\newcommand{\NN}{\mathbb{N}}
\newcommand{\RR}{\mathbb{R}}
\newcommand{\QQ}{\mathbb{Q}}

\newcommand{\edgep}[1]{~\overset{#1}{\text{------}}~}
\newcommand{\edgeg}{~\text{------}~}

\newcounter{sqindex}

\nolinenumbers

\title{Weighted Counting of Matchings
in Unbounded-Treewidth Graph Families}

\author{Antoine Amarilli}{LTCI, Télécom Paris, Institut Polytechnique de Paris,
France \and \url{https://a3nm.net/}}{a3nm@a3nm.net}{https://orcid.org/0000-0002-7977-4441}{Partially supported by the ANR project EQUUS
ANR-19-CE48-0019 and by the Deutsche Forschungsgemeinschaft (DFG, German Research Foundation) – 431183758.}
\author{Mikaël Monet}{Univ. Lille, Inria, CNRS, Centrale Lille, UMR 9189
CRIStAL, F-59000 Lille, France \and
\url{https://mikael-monet.net/}}{mikael.monet@inria.fr}{https://orcid.org/0000-0002-6158-4607}{}
\authorrunning{A.\ Amarilli and M.\ Monet}

\Copyright{Antoine Amarilli and Mikaël Monet}

\acknowledgements{We thank Florent Capelli and Sébastien Tavenas for
discussions on the problem, in particular on how to prove
Proposition~\ref{prp:fibo}.}

\supplementdetails[subcategory={Computer calculations}]{Software}{https://gitlab.com/Gruyere/supplementary-material-for-Weighted-Counting-of-Matchings}

\begin{document}
\maketitle

\begin{abstract}
We consider a weighted counting problem on matchings,
denoted~$\probmatch(\calG)$, on an arbitrary fixed graph family~$\mathcal{G}$.
The input consists of a graph~$G\in \mathcal{G}$ and of rational probabilities
of existence on every edge
of~$G$, assuming independence. The output is the probability of obtaining a \emph{matching}
of~$G$ in the resulting distribution, i.e., a set of edges that are pairwise disjoint.
It is known that, if~$\mathcal{G}$ has bounded \emph{treewidth}, then~$\probmatch(\calG)$ can be
solved in polynomial time. In this paper we show that, under some assumptions, bounded
treewidth in fact \emph{characterizes} the tractable
graph families for this problem. More precisely, we show intractability 
for all 
graph families $\calG$ satisfying 
the following
\emph{treewidth-constructibility} requirement:
given an integer~$k$ in unary, we can construct in
polynomial time a graph $G \in \calG$ with treewidth at least~$k$.  Our
hardness result is then the following: for \emph{any} treewidth-constructible
graph family~$\calG$, the problem $\probmatch(\calG)$ is intractable.
This generalizes known hardness results for weighted matching counting
under some restrictions that do not bound treewidth, e.g., being planar, 3-regular, or
bipartite; it also answers a question left open in~\cite{amarilli2016tractable}.
We also obtain a similar lower bound for the weighted counting of edge covers.

This is the full version of the article appearing in MFCS'22, containing
complete proofs.

\end{abstract}

\section{Introduction}
\label{sec:intro}
Many complexity results on computational problems
rely on a study of fundamental graph patterns such as independent sets, vertex
covers, edge covers, matchings, cliques, etc.
In this paper we specifically
study \emph{counting problems} for such patterns,
and 
for the most part
focus on counting the \emph{matchings}: given an input
graph~$G$,
we wish to count how many edge subsets of~$G$ are a matching, i.e., each vertex
has at most one incident edge.

Our goal is to address an apparent gap between the existing intractability and tractability
results for counting matchings and similar patterns.
On the one hand, counting the matchings is known to be \#P-hard, and hardness is
known even when the input graph is restricted in certain ways, e.g., being planar,
being 3-regular, or being
bipartite~\cite{vadhan2001complexity,greenhill2000complexity,xia2006regular,xia2007computational}. 
On the other hand, some restrictions can make the problem tractable, e.g.,
imposing that the input graphs have
\emph{bounded treewidth}~\cite{arnborg1991easy,amarilli2016tractable}, because
matchings can be described in monadic second-order logic.
But this does not settle the complexity of the problem;
could there be other restrictions on graphs that makes it tractable to count
matchings or other patterns?

This paper answers this question in the negative,
for a \emph{weighted} version of counting problems:
we show that,
at least for matchings and edge covers,
and
under a technical assumption on the graph family, 
the weighted 
counting problem is intractable 
if we do not bound the treewidth of the input graphs.
Thus, treewidth is the right parameter to ensure tractability.
Our weighted counting problems are of the following form:
we fix a graph family~$\mathcal{G}$ (e.g., $3$-regular graphs, graphs of treewidth
$\leq 2$),
we are given as input
a graph~$G$ of~$\calG$ along with an independent probability of existence 
for each edge,
and the goal is to compute the probability in this distribution of the subsets
of edges 
of~$G$ which have a certain property, e.g., 
they are a matching, they are an edge cover.
Note that the class~$\calG$ restricts the shape of the graphs, but the
edge probabilities are arbitrary --
and indeed there are known tractability results when we restrict the graphs and probabilities
to be symmetric~\cite{beame2015symmetric}.
Our paper shows the hardness of these problems when $\calG$ is not of bounded
treewidth;
the specific technical assumption on~$\calG$ is that one can effectively construct
graphs of~$\calG$ having arbitrarily high treewidth, i.e., the 
\emph{treewidth-constructible} requirement from~\cite{amarilli2016tractable}
(cf.\ Definition~\ref{def:tw-constr}):

\begin{result}
  \label{res:main}
  Let $\mathcal{G}$ be an arbitrary family of graphs which is
treewidth-constructible. Then the problem, given a graph $G=(V,E)$ of~$\mathcal{G}$
and rational probability values~$\pi(e)$ for every edge of~$G$, of computing the
  probability of a matching in~$G$ under~$\pi$, is \#P-hard under ZPP reductions.
\end{result}
We obtain an analogous result for edge covers.
Thus, as bounded-treewidth makes the problems tractable,
our results imply that 
treewidth characterizes the tractable graph families for these problems ---
for weighted counting, and assuming treewidth-constructibility. We leave open
the complexity of unweighted counting, and of weighted counting on graph families that have unbounded
treewidth but satisfy weaker requirements than treewidth-constructibility, e.g.,
being strongly unbounded
poly-logarithmically~\cite{kreutzer2010lower,ganian2014lower}.

The paper is devoted to showing Result~\ref{res:main} (with the proofs of
technical claims deferred to the appendix). 
At a high level, we use
the standard technique of reducing from the \#P-hard
problem of counting matchings on a $3$-regular planar
graph~$G$~\cite{xia2007computational}, using the randomized
polynomial-time grid minor extraction result of~\cite{chekuri2016polynomial} as
in~\cite{amarilli2016tractable}.
However, the big technical challenge is to reduce the counting of matchings
of~$G$ to the 
problem of computing the probability of a matching
on the arbitrary subdivision $G'$ of~$G$ that we
extract. For this, we use the classical interpolation method, where we design
a linear equation system relating the matchings to the result of polynomially
many oracle calls on~$G'$, with different probability assignments; and we argue
that the matrix is invertible. After the preliminaries
(Section~\ref{sec:prelims}), we present this proof, 
first in the case where $G'$ is a 6-subdivision of~$G$
(Section~\ref{sec:sub6}), and then when it is a $n$-subdivision, i.e., when all edges
are subdivided to the same length~$n$
(Section~\ref{sec:same-length}). 
These special cases already pose some difficulties, most of which are solved by
adapting techniques by Dalvi and Suciu~\cite{dalvi2012dichotomy}; e.g., to show
invertibility, we 
study the Jacobian determinant of the mapping associating edge probabilities to
the probability of matchings on paths with fixed endpoints, and we borrow a
technique from~\cite{dalvi2012dichotomy} to effectively construct
suitable rational edge probabilities.

The main novelties of this work are in Section~\ref{sec:general}, where
we extend the proof to the general case: $G'$ is a subdivision of~$G$,
and different edges of~$G$ may be subdivided in~$G'$ to different lengths.
To obtain the equation system, we show that we can assign
probabilities on short paths so that they ``behave'' like long
paths.
Proving this stand-alone \emph{emulation result} (Proposition~\ref{prp:fibo})
was the main
technical obstacle; the proof is by solving a system of
equations involving the Fibonacci sequence. It also introduces further
complications, e.g., dealing with numerical error (because the resulting
probabilities are irrational), and distinguishing even-length and odd-length
subdivisions.
After concluding the proof of Result~\ref{res:main} in
Section~\ref{sec:general},
we adapt it in Section~\ref{sec:others} to edge
covers.

\subparagraph*{Related work.}
Our work follows a line of results that show the intractability of some problems
on any ``sufficiently constructible'' unbounded-treewidth graph family.
Kreutzer and Tazari \cite{kreutzer2010lower} (see also~\cite{ganian2014lower}) show that there are formulas in an
expressive formalism (MSO$_2$) that are intractable to check on any
subgraph-closed unbounded treewidth graph family that is closed under taking
subgraphs and satisfies a requirement of being \emph{strongly unbounded
poly-logarithmically}. This was extended
in~\cite{amarilli2016tractable} to the weighted counting problem, this time for
a query in first-order logic, with a different hardness notion (\#P-hardness
under randomized reductions), and 
under the stronger requirement of
treewidth-constructibility. Our focus here is to show that the hardness of
weighted counting already holds for natural and well-studied graph properties,
e.g., ``being a matching''; this was left as an open problem
in~\cite{amarilli2016tractable}.

For such weak patterns, lower bounds were shown in~\cite{amarilli2016tractable} 
and~\cite{amarilli2019connecting} on the \emph{size of tractable
representations}: 
for any 
graph~$G$ of bounded degree
having treewidth~$k$, any so-called \mbox{\emph{d-SDNNF}} circuit representing the
matchings 
(or edge covers) 
of~$G$ must
have exponential size in~$k$.
However, this does not imply that the problems are intractable, as some
tractable counting algorithms
do not work via such circuit representations (e.g., the one
in~\cite{dalvi2012dichotomy}). Thus, our hardness result does not follow from
this size bound, but rather complements it.

The necessity of bounded treewidth has also been studied for graphical
models~\cite{chandrasekaran2008complexity} and Bayesian
networks~\cite{kwisthout2010necessity}. Specifically,
\cite{kwisthout2010necessity} shows the intractability of inference in a
Bayesian network as a function of the treewidth (but without otherwise
restricting the class of network), and \cite{chandrasekaran2008complexity} 
restricts the  shape of the
graphical model 
but allows arbitrary “potential functions” (whereas we assume
independence across edges).
There are also necessity results on treewidth for the problem of \emph{counting
the homomorphisms} between two structures in the CSP
context~\cite{dalmau2004complexity}; but this has no clear relationship to our problems,
where we do (weighted) \emph{counting of the substructures} that have a certain
form (e.g., are matchings).

Note that, unlike our problem of weighted counting of matchings, the problem of \emph{finding} a matching of maximal weight in a
weighted graph is tractable on arbitrary graphs, using Edmond's blossom
algorithm~\cite{blossom}.

\section{Preliminaries}
\label{sec:prelims}
We write~$\mathbb{N}^+$ for~$\mathbb{N} \setminus \{0\}$, and for $n\in
\mathbb{N}^+$ we write~$[n]$ the set~$\{0,\ldots,n-1\}$.  
We write $\RR$ the real numbers and $\QQ$ the rational numbers. Recall that
\emph{decimal fractions} are rational
numbers that can be written as a fraction 
$a/10^k$ of an integer $a$ and a power
of ten~$10^k$.

\subparagraph*{Reductions and complexity classes.}
Recall that $\shp$ is the class of counting problems that count
the number of accepting paths of a nondeterministic polynomial-time Turing
machine. A problem $P_1$ is \emph{\shp-hard} if every problem $P_2$ of \shp~reduces to $P_1$ in
polynomial time; following Valiant~\cite{valiant1979complexity,V79}, we use here the notion of \emph{Turing reductions}, i.e., $P_2$ can
be solved in polynomial time with an oracle for~$P_1$.
We specifically study what we call \emph{\shp-hardness under
zero-error probabilistic polynomial-time (ZPP) reductions}. 
To define these, we define a \emph{randomized algorithm} as an algorithm that
has access to an additional random tape.
We say that a decision problem is in \emph{ZPP} 
if there is a randomized algorithm that (always) runs
in polynomial time on the input instance, and 
returns the correct answer on the instance (i.e., accepting or rejecting) with
some constant probability, and
otherwise returns a special failure value. The probabilities are taken over
the draws of the contents of the random tape. The exact value of the acceptance
probability is not important, because we can make it exponentially small by
simply repeating the algorithm polynomially many times.
Going beyond decision problems, a \emph{ZPP algorithm} is a randomized algorithm that runs in polynomial time but may return
a special failure value with some constant probability.
A \emph{ZPP (Turing) reduction}
from a problem~$P_1$ to a problem~$P_2$ is then a ZPP algorithm having access to
an oracle for~$P_2$ that
takes an instance of problem~$P_1$, runs in polynomial time, returns the
correct output (for $P_1$) with some constant probability, and returns
the special failure value otherwise. Again, the failure probability can be made arbitrarily small by
invoking the reduction multiple times.
A problem~$P_2$ is then said to be \emph{\#P-hard under ZPP reductions}
if any \#P-hard problem~$P_1$ has a ZPP reduction to it. We will implicitly rely on
the fact that we can show \#P-hardness under ZPP reductions by reducing in ZPP
from any problem which is \#P-hard (under Turing reductions); see
Appendix~\ref{apx:formal} for details.

\subparagraph*{Graphs and problem studied.}
A finite undirected \emph{graph}~$G=(V,E)$
consists of a finite set~$V$ of \emph{vertices} (or \emph{nodes}) and of 
a set~$E$ of \emph{edges} of the form~$\{x,y\}$ for~$x,y \in V$ with $x\neq
y$.  
A \emph{graph family} $\mathcal{F}$ is a (possibly infinite) set of graphs.
For~$v\in V$, we write~$\mathcal{E}_G(v)$
for the set of edges that are incident to~$v$. 
Recall that a \emph{matching} of $G$ is a set of edges~$M \subseteq E$ that
do not share any vertices, i.e., for every~$e,e'\in M$ with~$e\neq
e'$ we have $e \cap e' = \emptyset$; or equivalently, we have~$|\{\mathcal{E}_G(v)\cap M\}| \leq 1$ for all~$v\in V$.
For a graph family $\mathcal{F}$, we write~$\shpmatch(\mathcal{F})$ the problem of counting the matchings for graphs
in~$\mathcal{F}$: the input is a graph $G \in \mathcal{F}$, and
the output is the number of matchings of~$G$, written $\shpmatch(G)$.

We study a weighted version of $\shpmatch$, defined on \emph{probabilistic
graphs}.
A \emph{probabilistic graph} is a
pair $(G,\pi)$ where~$G = (V, E)$ is a graph and
$\pi\colon
E\to [0,1]$ maps every edge~$e$ of~$H$ to a probability value~$\pi(e)$.
The probabilistic graph $(G, \pi)$ defines a probability distribution on the set of
subsets $E'$ of $E$, where each edge $e \in E$ is in $E'$ with probability
$\pi(e)$, assuming independence across edges. Formally, the probability of each
subset~$E'$
is:
\[\Pr_{G,\pi}(E') \colonequals \prod_{e \in E'} \pi(e)
\times \prod_{e \in E \setminus E'} (1-\pi(e)).\]

Given a probabilistic
graph~$(G,\pi)$, the \emph{probability of a matching in~$G$ under~$\pi$},
denoted $\prm(G,\pi)$, is the probability of obtaining a matching in the
distribution. Formally:
\begin{equation}
\label{eqn:match}
\prm(G,\pi)\ \colonequals
\sum_{\text{matching }M\text{ of }G} \Pr_{G,\pi}(M).
\end{equation}
In particular, if~$\pi$ maps every edge to the probability~$1/2$, then we have
$\prm(G,\pi) = \shpmatch(G)/2^{|E|}$.
For a graph family $\calF$, we will study the problem $\probmatch(\calF)$ of
computing the probability of a matching: the input is a probabilistic graph $(G,
\pi)$ where~$G \in \calF$ and $\pi$ is an arbitrary function with rational
probability values, and the output is $\prm(G, \pi)$. Note that $\calF$ only specifies the
graph $G$ and not the probabilities $\pi$, in particular $\pi$ can give
probability~$0$ to edges, which amounts to removing them.

\subparagraph*{Treewidth and topological minors.} 
Treewidth is a parameter mapping any 
graph~$G$ to a number
$\tw(G)$ intuitively describing how far $G$ is
from being a tree.
We omit the formal definition of treewidth (see~\cite{treewidth}), as
we only rely on the following \emph{extraction result}:
given any \emph{planar graph~$H$ of maximum degree~$3$}, and a graph~$G$ 
of \emph{sufficiently high}
treewidth, it is possible (in randomized polynomial time) to find~$H$ as a
\emph{topological minor} of~$G$.  We now define this.

The \emph{degree} of a node~$v$ in~$H=(V_H,E_H)$ is 
simply~$|\mathcal{E}_G(v)|$.
We say that $H$ is \emph{$3$-regular} if every vertex has degree~$3$, and call
$H$ \emph{planar} if it can be drawn on the plane without
edge crossings, in the usual sense~\cite{planar}.  
Given~$H$ and $\eta\colon E_H \to \mathbb{N}^+$,
the \emph{$\eta$-subdivision
of~$H$}, written $\sub(H,\eta)$, is the graph obtained from~$H$
by replacing every edge~$e = \{x,y\}$ by a path of length~$\eta(e)$,
whose end vertices are identified with~$x$ and~$y$, all intermediate vertices being
fresh across all edges.
We abuse notation and write $\sub(H,i)$ for $i \in \mathbb{N}_+$ to mean $\sub(H, \eta_i)$
for $\eta_i$ the constant-$i$ function.
Note that $\sub(H, 1) = H$.
A \emph{subgraph} of a graph~$G = (V_G, E_G)$ is a graph $(V_G', E_G')$ where
$E_G' \subseteq E_G$ and $V_G' \subseteq V_G$ such that $e \subseteq V_G'$ for
each edge $e \in E_G'$.
The graph 
$H=(V_H,E_H)$ is a \emph{topological minor}
of the graph~$G=(V_G,E_G)$
if there is a function~$\eta\colon E_H
\to \mathbb{N}^+$ such that there is an isomorphism $f$ from the subdivision
 $\sub(H, \eta) = (V_H', E_H')$ to some subgraph $G' = (V_G', E_G')$ of~$G$, i.e., 
a bijection $f\colon V_H' \to V_G'$ such that for every $x, y \in V_H'$ we have
$\{x,y\}\in E_H'$ if
and only if~$\{f(x),f(y)\}\in E_G'$.

We can now state the extraction result that we use, which follows from the work
of Chekuri and Chuzhoy~\cite{chekuri2016polynomial}:

\begin{theorem}[Direct consequence of \cite{chekuri2016polynomial}, see, e.g.,
  \cite{amarilli2016tractable}, Lemma~4.4]
\label{thm:top-min}
There exists~$c\in \mathbb{N}$ and 
  a ZPP algorithm\footnote{The randomized algorithm from~\cite{chekuri2016polynomial} is
  indeed a ZPP algorithm because the output that it returns (namely, a prospective
embedding of a grid as a topological minor of the input graph) can be verified in
(deterministic) polynomial time. Hence, we can always detect when the algorithm
  has failed, and then return the special failure value.}
  that, given as input a planar graph~$H=(V_H,E_H)$ of maximum degree~$3$
and another 
  graph~$G$ with $\tw(G) \geq
|V_H|^c$, computes 
a subgraph $G'$ of~$G$,
  a function $\eta:V_H \to
  \mathbb{N}^+$, and an isomorphism from~$\sub(H,\eta)$ to $G'$ (witnessing
  that~$H$ is a topological minor of~$G$).
\end{theorem}

Our intractability result will apply to graph families where large treewidth graphs can be
efficiently found, which we formalize as
\emph{treewidth-constructibility} like in~\cite{amarilli2016tractable}:

\begin{definition}
\label{def:tw-constr}
A graph family~$\mathcal{F}$ is \emph{treewidth-constructible}
if there is a polynomial-time algorithm that, given 
  an integer~$k$
  written in unary\footnote{Note that the existence of such an algorithm for $k$
  written in unary would be implied by the same claim but with $k$ given in binary. In other words, the existence of an algorithm for $k$
given in unary is a weaker requirement. This is simply because, given an integer in
unary, we can convert it in PTIME to an integer in binary.},
outputs a graph~$G\in \mathcal{F}$ with~$\tw(G)\geq k$.
\end{definition}

\subparagraph*{Kronecker products and Vandermonde matrices.}
To simplify notation, we will work with matrices indexed with arbitrary
finite sets (not necessarily ordered).  Given two finite sets~$I,J$ 
of same cardinality,
we write~$\mathbb{R}^{I,\, J}$ (resp.,~$\mathbb{Q}^{I,\, J}$)  the set of matrices with real values (resp., rational values) 
whose rows are indexed by~$I$ and columns by~$J$. When $\bm{A} \in
\mathbb{R}^{I,\, J}$ and $(i,j) \in I \times J$, we write~$a_{i,j}$ the
corresponding entry. 
We recall that the inverse of an
invertible matrix $M$ with entries in $\mathbb{Q}$ also has entries in $\mathbb{Q}$
and can be computed in polynomial time in the encoding size of~$M$.

Given two matrices $\mA \in \mathbb{R}^{I,\, J}$ and $\mB
\in \mathbb{R}^{K,\, L}$, the \emph{Kronecker product of~$\mA$ and~$\mB$},
denoted $\mA \otimes \mB$, is the matrix~$\mC \in \mathbb{R}^{I\times K,\,
\,J\times L}$ defined by~$c_{(i,k),(j,l)} \colonequals a_{i,j} \times b_{k,l}$
for $(i,j,k,l)\in I\times J\times K \times L$.  Recall that~$\mA \otimes \mB$
is invertible if and only if both~$\mA$ and~$\mB$ are. 
For $n \in \mathbb{N}^+$ and $(p_0,\ldots,p_{n-1}) \in \mathbb{R}^n$, we denote
by~$\mathcal{V}(p_0,\ldots,p_{n-1})$ the \emph{Vandermonde matrix with
coefficients $(p_0,\ldots,p_{n-1})$}, i.e., the matrix in~$\mathbb{R}^{[n], [n]}$
whose~$(i,j)$-th entry is~$p_i^j$. Recall that
this matrix
is invertible if and only if
the $p_0, \ldots, p_{n-1}$ are pairwise distinct.

\section{Proof When Every Subdivision Has Length 6}
\label{sec:sub6}
Towards showing our main result (Result~\ref{res:main}), we first show in this
section a much simpler result: counting the matchings of a graph $G$ reduces to
counting the probability of a matching on the graph where each edge is 
subdivided into a path of length~$6$.
We use similar techniques to previous work, in particular
Greenhill~\cite{greenhill2000complexity} and Dalvi and
Suciu~\cite{dalvi2012dichotomy}, but present them in detail
because we will adapt them in the rest of the paper.
Formally, in this section, we show:

\begin{proposition}
\label{prp:sub6}
For any graph family~$\mathcal{F}$, the problem~$\shpmatch(\mathcal{F})$
reduces in polynomial time to~$\probmatch(\mathcal{G})$
where $\mathcal{G} = \{\sub(H, 6) \mid H \in \mathcal{F}\}$.
\end{proposition}

Let~$H=(V,E)$ be a graph in~$\calF$ for which we wish to count the number
of matchings, with $m\colonequals |E|$.  Let us start by fixing for the remainder
of this section an arbitrary orientation~$\overrightarrow{H}$ of~$H$ obtained
by choosing some orientation of the edges, i.e.,
$\overrightarrow{H}=(V,\overrightarrow{E})$ is a \emph{directed} graph
where for every edge~$\{x,y\}\in E$ we add exactly one
of~$(x,y)$ or~$(y,x)$ in~$\overrightarrow{E}$.
The high-level idea of the reduction is then the following. First, 
using~$\overrightarrow{H}$, we 
define some sets~$S_\tau$,
based on $4$-tuples
$\tau \in [m+1]^4$,
such that the number of matchings of~$H$
can be computed from the cardinalities~$|S_\tau|$.
Second, we argue that these
cardinalities 
can
be connected to the results of oracle calls for the PrMatching problem by a
system of linear equations. Third, we argue that the matrix of this system can
be made invertible.
We now detail these three steps.

\subparagraph*{Step 1: Defining the sets~$S_\tau$ and linking them to matchings.}
We define a \emph{selection function} of the graph~$H$
as a function~$\mu$ that maps each vertex~$x \in V$ to at most one incident edge,
i.e., to a subset of~$\mathcal{E}_H(x)$ of size at most one.
We will partition the set of selection functions by 
counting the number of edges 
of each \emph{type} that 
each selection function has, as defined next.
Given a selection function~$\mu$, consider each edge~$e = (x, y)$ of~$\overrightarrow{H}$. The edge~$e$ can have one of four types:
letting $b$ be~$1$ if~$\mu(x)$ selects~$e$ (i.e., $\mu(x) = \{\{x, y\}\}$) and~$0$ otherwise (i.e., 
$\{x,y\} \notin \mu(x)$),
and letting $b'$ be~$1$ if $\mu(y)$ selects~$e$ 
and $0$ otherwise, we say that~$e$ \emph{has type $bb'$ with respect to (w.r.t.)~$\mu$}.
We now define the sets~$S_\tau$ as follows.
\begin{definition}
\label{def:Staudef}
For a 4-tuple $\tau \in [m+1]^4$,
indexed in binary,
let
$S_\tau \subseteq S$ 
be 
the
set of the selection functions~$\mu$ such that,
 for all $b,b' \in \{0, 1\}$,
  precisely $\tau_{bb'}$ edges
  have type~$bb'$ w.r.t.~$\mu$.
\end{definition}
Observe that~$S_\tau$ is empty unless
$\tau_{00}+\tau_{01}+\tau_{10}+\tau_{11} = m$.
We can then easily connect the cardinalities~$|S_\tau|$ to the number of matchings
of~$H$ as follows (see Appendix~\ref{apx:param}):

\begin{restatable}{fact}{fctparam}
\label{fct:param}
  We have that
$\shpmatch(H) = \sum_{\substack{\tau \in [m+1]^4 \\ \tau_{01} = \tau_{10} = 0}} |S_\tau|$.
\end{restatable}

\subparagraph*{Step 2: Recovering the $|S_{\tau}|$ from oracle calls.}
We now explain how to use the oracle for $\probmatch(\mathcal{G})$ to
compute in polynomial time all the values $|S_{\tau}|$, allowing us to
conclude via Fact~\ref{fct:param}.
We will invoke the oracle on 
$(m+1)^4$ probabilistic graphs, denoted~$H_6(\kappa)$
for~$\kappa\in [m+1]^4$, as defined next.
To this end, let us consider~$(m+1)^4$ 4-tuples of probability values, written $\rho_\kappa = (\rho_{\kappa,00},
\rho_{\kappa,01},
\rho_{\kappa,10},\rho_{\kappa,11})\in [0,1]^4$ for $\kappa \in [m+1]^4$;
the precise choice of these values will be explained in Step 3.
For~$\kappa\in [m+1]^4$, we then define $H_6(\kappa)$ to be the probabilistic graph $(H_6, \pi_\kappa)$ 
where $H_6 \colonequals \sub(H, 6)$ is the 6-subdivision of~$H$ 
and the probabilities
$\pi_\kappa$ are defined as follows.
For every 
directed edge~$(x,y)$ of $\overrightarrow{H}$,
the subdivision~$H_6$ contains an
(undirected) path between~$x$ and~$y$, and we define $\pi_\kappa$ on this path
as follows:
\[
  \hfill x \edgep{1/2} v_1 \edgep{\rho_{\kappa,00}} v_2 \edgep{\rho_{\kappa,01}} v_3
  \edgep{\rho_{\kappa,10}} v_4 \edgep{\rho_{\kappa,11}}
  v_5 \edgep{1/2} y \hfill\]

We now introduce some notation for the probability of matchings in paths
of length~4. 
We write $\Pi_4(\rho_\kappa)$ the probability of having a matching in the
4-edge path with successive probabilities $\rho_{\kappa,00}, \rho_{\kappa,01},
\rho_{\kappa,10}, \rho_{\kappa,11}$.
The value can be explicitly computed as a polynomial in the values
$\rho_{\kappa,bb'}$, e.g., using Equation~\eqref{eqn:match}.
Accordingly, we will also use $\Pi_4$ as a
polynomial with real variables, i.e., $\Pi_4(\chi)$ for a 4-tuple $\chi$ of real
values (which may not be in $[0,1]$).
We 
also define variants of these definitions that account for the two surrounding edges,
i.e., those with probability $1/2$: for $b,b'\in\{0,1\}$, write
$\Pi_4^{bb'}(\rho_\kappa)$ to denote the probability of having a matching in the
same 4-edge path 
but when adding an edge incident to the first vertex with probability~$1$ if $b=1$, and adding an
edge incident to the last vertex with probability~$1$ if~$b' = 1$. Equivalently, 
$\Pi_4^{bb'}(\rho_\kappa)$ is the 
probability of obtaining a matching where we further require if $b=1$ that the
edge with probability $\rho_{\kappa,00}$ is \emph{not} taken, and if~$b'=1$ that
the edge with probability $\rho_{\kappa,11}$ is \emph{not} taken. 
The values $\Pi_4^{bb'}(\rho_\kappa)$
are also explicitly computable as polynomials, and again we also see
$\Pi_4^{bb'}$ as a polynomial with real variables.
To simplify notation, for~$b,b'\in \{0,1\}$ and $\kappa\in[m+1]^4$ let us
write~$\Lambda_{\kappa,bb'} \colonequals \Pi^{bb'}_{4}(\rho_\kappa)$.

We 
then show that the probability of a matching in the subdivided
graph $H_6$ can be obtained by first summing over the possible edge type cardinalities $\tau$,
and then 
regrouping the edges of the same type by noticing that
the 
matchings
corresponding
to the selection functions in the set $S_\tau$ all have the same probability.
Namely, we show
(cf.\ Appendix~\ref{apx:match-to-prob}):

\begin{restatable}{fact}{matchtoprob}
\label{fact:match-to-prob}
  For each $\kappa \in [m+1]^4$, we have:
\begin{equation*}
2^{2m} \times \prm(H_6(\kappa)) \ = \  \sum_{\tau\in [m+1]^4} |S_{\tau}|
  \times (\Lambda_{\kappa,00})^{\tau_{00}} \times (\Lambda_{\kappa,01})^{\tau_{01}} \times
  (\Lambda_{\kappa,10})^{\tau_{10}} \times (\Lambda_{\kappa,11})^{\tau_{11}}.
\end{equation*}
\end{restatable}
 
Now, let us write $c_{\kappa} \colonequals \prm(H_6(\kappa))$ the value returned
by the oracle call on~$H_6(\kappa)$, 
and let~$\mC$ be the vector of these oracle answers.
Let~$\mS$ be the vector $|S_{\tau}|$ of the values that we wish to compute.
Both these vectors are indexed by $[m+1]^4$.
Observe that the equation above
defines a
system of linear equations~$\mV \mS = \mC$ with
$\mV \in \mathbb{R}^{[m+1]^4, [m+1]^4}$ defined by
\[
  v_{\kappa,\tau} \ \colonequals \ 2^{-2m} 
  \times (\Lambda_{\kappa,00})^{\tau_{00}} \times (\Lambda_{\kappa,01})^{\tau_{01}}
  \times (\Lambda_{\kappa,10})^{\tau_{10}} \times (\Lambda_{\kappa,11})^{\tau_{11}}.
\]
Therefore, 
if we can choose 4-tuples of probability values~$\rho_\kappa$ that make~$\mV$ invertible,
we would be able to recover all~$|S_{\tau}|$ values from the oracle answers $\mC$, from which we could compute
the number of matchings of~$H$ using Fact~\ref{fct:param}. This is what we do next.

\subparagraph*{Step 3: Making $\mV$ invertible.}
We now explain how to choose in polynomial time $(m+1)^4$ 4-tuples~$\rho_\kappa$ 
of rational probability values, for $\kappa \in [m+1]^4$,
such that $\mV$ 
is invertible. 
To this end, consider the matrix $\mM$ defined like $\mV$ except that
each 4-tuple~$\rho_\kappa$ is replaced by a 4-tuple of variables
$\chi_\kappa = (\chi_{\kappa,00},
\chi_{\kappa,01},
\chi_{\kappa,10},\chi_{\kappa,11})$. Each cell~$m_{\kappa,\tau}$ of~$\mM$ is
then a polynomial
$P_\tau$ in the 4 variables $\chi_{\kappa,bb'}$ for $b,b'\in\{0,1\}$; 
in particular, note that the polynomial only depends on the column $\tau$, whereas
the variables $\chi_{\kappa,bb'}$ only depend on the row~$\kappa$.
We can then find suitable values $\rho_\kappa$ using a technique 
introduced by Dalvi and
Suciu~\cite{dalvi2012dichotomy} (see Appendix~\ref{apx:85}):

\begin{proposition}[From Proposition 8.44 of \cite{dalvi2012dichotomy}]
  \label{prp:85}
  Fix $k \in \NN$, 
  let $(x_i)_{i \in I}$ be $k$-tuples of real variables indexed by a finite
  set~$I$,
  let $(P_j)_{j \in J}$ be polynomials in $k$ variables indexed by a finite set~$J$, and
  consider the matrix $\mMM$ indexed by~$I \times J$ such that $m_{i,j} =
  P_j(x_i)$ for all $(i,j)\in I\times J$.
  Assume that $\det(\mMM)$ is not the null polynomial.
  There is an algorithm that runs in polynomial time in $\mMM$ 
  and finds $|I|$ $k$-tuples of decimal fractions $(a_i)_{i\in I}$ with values in $[0, 1]$
  such that the matrix obtained by substituting each $x_i$ by~$a_i$ in~$\mMM$ is
  invertible.
\end{proposition}

If $\det(\mM)$ 
is not the null polynomial, we can invoke this result with $k =
4$ and~$I=J=[m+1]^4$ on the matrix~$\mM$, which gives us in polynomial time the desired
rational probability values~$\rho_\kappa$ (namely, the~$a_i$ from the
proposition) and concludes the
proof of Proposition~\ref{prp:sub6}.

Hence, the only remaining point is to argue that 
$\det(\mM)$ is not the null polynomial (in the~$\chi_\kappa$). 
To this end, let us study the mapping 
$\xi:\mathbb{R}^4 \to \mathbb{R}^4$, defined as follows, with $\chi$ denoting a 4-tuple of real variables:
$\xi(\chi)  \colonequals \left(
\Pi_4^{00}(\chi),\,
\Pi_4^{01}(\chi),\, 
\Pi_4^{10}(\chi),\,
\Pi_4^{11}(\chi)\right)$.
For a 4-tuple of reals $\rho$, we call the mapping $\xi$ \emph{invertible}
around point $\rho$ if there is $\epsilon>0$ such that the
\emph{$\epsilon$-neighborhood around~$\xi(\rho)$}, i.e., the set  $\{\alpha \in
\mathbb{R}^4\, \mid \,|\alpha_{bb'}-\xi(\rho)_{bb'}| \leq \epsilon \text{ for
each }b,b' \in \{0,1\}\}$, is
included in
the image of~$\xi$.
We conclude by showing two claims:

\begin{fact}
  The mapping $\xi$ is invertible around some point.
\end{fact}

\begin{proof}
  By the inverse function theorem~\cite{inversefunction}, if the \emph{Jacobian
  determinant} of~$\xi$ at a point is not null,
  then~$\xi$ is invertible around that point.
Recall that the Jacobian
determinant of~$\xi$ is the determinant of the \emph{Jacobian matrix} of~$\xi$,
which is the $4\times 4$ matrix~$\mJ_\xi$ whose entry at cell~$((b_1,b_2),(b'_1,b'_2))$ is
$\frac{\partial \Lambda_{\chi,b_1 b_2}}{\partial \chi_{b'_1 b'_2}}$.
  We explicitly compute $\det(\mJ)$ with the help of SageMath, showing that it
  is not the null polynomial
  (see Appendix~\ref{apx:jacobian-first}).
\end{proof}

\begin{fact}
  If $\xi$ is invertible around some point $\rho$,
then $\det(\mM)$ is not the null polynomial.
\end{fact}

\begin{proof}
  The invertibility of~$\xi$ around~$\rho$ implies that there exist, for each $b,b'\in\{0,1\}$, a set of $m+1$ distinct values
$\Psi_{bb'} \colonequals \{\psi_{bb',0}, \ldots, \psi_{bb',m-1}\}$ such that 
the Cartesian product
$\Psi \colonequals \bigtimes_{b,b'\in \{0,1\}} \Psi_{bb'}$ is included in the~$\epsilon$-neighborhood of~$\xi(\rho)$.
Let us index the $(m+1)^4$ 4-tuples of $\Psi$ as $\psi_\kappa$ for $\kappa \in [m+1]^4$,
i.e., $\psi_\kappa = (\psi_{00,\kappa_{00}}, \psi_{01,\kappa_{01}},\psi_{10,\kappa_{10}},\psi_{11,\kappa_{11}})$.
Using invertibility, let 
$\alpha_\kappa$ 
be a preimage of each $\psi_\kappa$, i.e., 
$\xi(\alpha_\kappa) = \psi_\kappa$ 
for all $\kappa \in [m+1]^4$. But then observe that, for this choice of 
$\chi_\kappa$ (i.e., substituting the~$\chi_\kappa$ by the~$\alpha_\kappa$), 
each cell $u_{\kappa,\tau}$ of the matrix $\mM$ becomes:
\[
  u_{\kappa,\tau} = 2^{-2m} \times (\psi_{00,\kappa_{00}})^{\tau_{00}} \times (\psi_{01,\kappa_{01}})^{\tau_{01}} \times (\psi_{10,\kappa_{10}})^{\tau_{10}} \times (\psi_{11,\kappa_{11}})^{\tau_{11}}.
\]
Thus, $\mM$ is the Kronecker product of four Vandermonde matrices $\mM_{bb'}$ for
$b,b'\in\{0,1\}$, where~$\mM_{bb'}$ is
$\mathcal{V}(\psi_{bb',0},\ldots,\psi_{bb',m-1})$.
As the $\Psi_{bb'}$ consist of pairwise distinct values, these 
  Vandermonde 
  matrices are invertible, and their Kronecker product $\mM$ also is.
\end{proof}

\section{Proof When All Subdivisions Have the Same Length~$\geq 7$}
\label{sec:same-length}
We now prove a 
variant of Proposition~\ref{prp:sub6} where all edges of
the initial graph are subdivided the same number of times (at least~$7$). Given a
graph $H$ and integer $K > 0$, we write~$G_K$ to mean $\sub(H, K)$. In this section we 
show:

\begin{proposition}
\label{prp:same-length}
Fix an integer~$K \geq 7$. Then, for any graph family~$\mathcal{F}$,  the
  problem
$\shpmatch(\mathcal{F})$ reduces in polynomial time
to~$\probmatch(\calG)$,
where $\mathcal{G} = \{H_K \mid H \in \mathcal{F}\}$.
\end{proposition}

To prove this, we follow the same strategy as for
Proposition~\ref{prp:sub6}.  The first step --- the definition of the~$S_\tau$
--- is
strictly identical; for $m$ the number of edges of~$H$, we fix again an orientation~$\overrightarrow{H}$ 
of~$H$, and denote~$S_\tau$ for~$\tau \in [m+1]^4$ the~$(m+1)^4$ sets of
selection functions defined
from~$\overrightarrow{H}$ as in
Definition~\ref{def:Staudef}.
In particular, Fact~\ref{fct:param} still holds.
Now, we will again construct $(m+1)^4$ probabilistic
graphs, denoted~$H_K(\kappa)$ for~$\kappa\in [m+1]^4$, such that,
letting~$c_{\kappa} \colonequals \prm(H_K(\kappa))$, the
$|S_{\tau}|$ and the~$c_{\kappa}$ 
form a linear system of
equations~$\mV \mS = \mC$. We will then again use the Jacobian technique to argue
that the determinant of this matrix is not the null polynomial, and complete the proof using 
Proposition~\ref{prp:85}
to compute in polynomial time rational
values that make~$\mV$ have rational entries and be invertible.  The difference with Section~\ref{sec:sub6}
is in the construction of the probabilistic graphs~$H_K(\kappa)$, and in the
Jacobian determinant.
Before we start, we need to extend the notation from Section~\ref{sec:sub6}.

\subparagraph*{Probabilistic path graphs.}
For $n\in \mathbb{N}^+$ we denote by~$P_n$ the \emph{path of
length~$n$}, i.e., $P_n = (\{v_0, \ldots, v_n\}, E)$ where~$E = \{\{v_i, v_{i+1}\}
\mid 0 \leq i \leq n-1\}$.
For~$\rho  \in [0,1]^n$, we let~$P_n(\rho)$ be the probabilistic graph 
where each edge $\{v_i, v_{i+1}\}$ of~$P_n$ has probability $\rho_i$.  We
write $\Pi_n(\rho)$ the probability of a matching in $P_n(\rho)$.  For~$b,b' \in \{0, 1\}$, we write $\Pi^{bb'}_{n}(\rho)$ to denote $\Pi_{n+2}(b, \rho, b')$, i.e., the
probability of a matching in $P_n(\rho)$ where we add an
edge to the left if $b=1$ and add an edge to the right if $b'=1$.
In
particular $\Pi^{00}_{n}(\rho) = \Pi_n(\rho)$.  We call
the quadruple of values $\Pi_n^{bb'}(\rho)$ for~$b,b' \in
\{0,1\}$ the \emph{behavior} of the path $P_n(\rho)$. 
Each $\Pi_n^{bb'}(\rho)$ is a polynomial in the probabilities $\rho$, 
and thus we also see $\Pi_n^{bb'}$ as a polynomial with real variables as
in Section~\ref{sec:sub6}.
We will use the following two lemmas. The first one
 expresses the behavior of the concatenation of
two paths as a function of the behavior of each path (cf.\
Appendix~\ref{apx:concat}):

\begin{restatable}{lemma}{concat}
\label{lem:concat}
Let~$n,n' \in \mathbb{N}^+$ and~$\rho \in [0,1]^n$, $\rho' \in [0,1]^{n'}$ be tuples of probability values. 
Then, for every~$b,b' \in \{0,1\}$, we have:
\begin{align*}
  \Pi_{n+n'}^{bb'}(\rho, \rho') \ = \ & (\Pi_n^{b0}(\rho) \times
  \Pi_{n'}^{1b'}(\rho'))
  + (\Pi_n^{b1}(\rho) \times \Pi_{n'}^{0b'}(\rho'))
  - (\Pi_n^{b1}(\rho) \times \Pi_{n'}^{1b'}(\rho')).
\end{align*}
\end{restatable}

The second lemma expresses
the values $\Pi^{bb'}_{n}(1/2,\ldots,1/2)$ in terms of the \emph{Fibonacci
sequence}.
Recall that this is the integer sequence defined by
$f_0 \colonequals 0$, $f_1 \colonequals 1$, and $f_n \colonequals f_{n-1} +
f_{n-2}$ for all $n \in \mathbb{N}^+$, and that this sequence
satisfies \emph{Cassini's
identity}~\cite{cassini}, which says that
$f_n^2 = f_{n+1} f_{n-1} + (-1)^{n+1}$
for every~$n \in \mathbb{N}^+$.
We have (cf.\ Appendix~\ref{apx:fiboexpr}):

\begin{restatable}{lemma}{fiboexpr}
\label{lem:fiboexpr}
For all $n \in \mathbb{N}^+$, $b, b' \in \{0, 1\}$, we have
$\Pi^{bb'}_{n}(1/2,\ldots,1/2) = \frac{f_{n+2-b-b'}}{2^n}$.
\end{restatable}

\subparagraph*{Proving Proposition~\ref{prp:same-length}.}
Let us now build the graphs~$H_K(\kappa)$. As before, consider~$(m+1)^4$
4-tuples of probability values~$\rho_\kappa = (\rho_{\kappa,00},
\rho_{\kappa,01},
\rho_{\kappa,10},\rho_{\kappa,11})$ for~$\kappa \in
[m+1]^4$, to be chosen later.
Each graph $H_K(\kappa)$ has $H_K$ as its
underlying graph,
and for every directed edge~$(x,y)\in
\overrightarrow{H}$, we set the probabilities on the corresponding undirected path in~$H_K$ as
follows:
\[
  x \edgep{1/2} v_1 \edgep{\rho_{\kappa,00}} v_2 \edgep{\rho_{\kappa,01}} v_3 \edgep{\rho_{\kappa,10}} v_4 \edgep{\rho_{\kappa,11}}
  v_5 \edgep{1/2} v_6 \edgep{1/2} \cdots \edgep{1/2} v_{K-1} \edgep{1/2} y\]
Note that this is like in Section~\ref{sec:sub6}, but giving probability
$1/2$ to the
$N\colonequals K-6$
extra edges on the path.
For~$b,b'\in \{0,1\}$ we write again
$\Lambda_{\kappa,bb'} \colonequals \Pi^{bb'}_{4}(\rho_\kappa)$ the behavior of
the 4-path with probabilities $\rho_\kappa$,
and we define the behavior $\Upsilon_{\kappa,bb'} \colonequals
\Pi^{bb'}_{K-2}(\rho_\kappa,1/2,\ldots,1/2)$ of 
the path depicted above
without the
first and last edges.
Note that with Lemma~\ref{lem:concat} and
Lemma~\ref{lem:fiboexpr},
we can then express the~$\Upsilon_{\kappa,bb'}$ as a function of
the~$\Lambda_{\kappa,bb'}$ and of the Fibonacci numbers:
\begin{fact}
\label{fact:alpha-t-f}
  We have $\Upsilon_{\kappa,bb'} = 2^{-N} \times (\Lambda_{\kappa,b0} \times f_{N+1-b'}\, +\, \Lambda_{\kappa,b1}\times f_{N-b'}) $ for~$b,b' \in \{0,1\}$.
\end{fact}

Studying the graphs $H_K(\kappa)$, by the same reasoning as 
for Fact~\ref{fact:match-to-prob},
we can easily show:
\begin{equation}
\label{eqn:match-to-prob-same}
2^{2m} \times \prm(H_K(\kappa)) \ = \  \sum_{\tau\in [m+1]^4} |S_{\tau}| \times (\Upsilon_{\kappa,00})^{\tau_{00}} \times (\Upsilon_{\kappa,01})^{\tau_{01}} \times (\Upsilon_{\kappa,10})^{\tau_{10}} \times (\Upsilon_{\kappa,11})^{\tau_{11}}.
\end{equation}
This is again a system of linear equations $\mV \mS = \mC$ with $\mV
\in \mathbb{R}^{[m+1]^4, [m+1]^4}$, where 
$v_{\kappa,\tau} \ \colonequals 2^{-2m} \times (\Upsilon_{\kappa,00})^{\tau_{00}} \times (\Upsilon_{\kappa,01})^{\tau_{01}} \times (\Upsilon_{\kappa,10})^{\tau_{10}} \times (\Upsilon_{\kappa,11})^{\tau_{11}}$.  
To show
that we can compute in polynomial time 4-tuples of rational probability values $\rho_\kappa$
 for~$\kappa \in [m+1]^4$ 
that make~$\mV$ have rational entries and be invertible, we reason 
as in
Section~\ref{sec:sub6}.
Specifically, we
study
the Jacobian determinant of the
mapping
$\xi_N:\chi \mapsto \big(\Pi^{00}_{K-2}(\chi,1/2,\ldots,1/2),\,
\Pi^{01}_{K-2}(\chi,1/2,\ldots,1/2),\linebreak
\Pi^{10}_{K-2}(\chi,1/2,\ldots,1/2),\,
\Pi^{11}_{K-2}(\chi,1/2,\ldots,1/2)\big)$,
where~$\chi$ is a 4-tuple of real variables. We show that this
determinant
is not the null polynomial.
To do this, starting from the Jacobian~$\mJ_\xi$ of Section~\ref{sec:sub6}, 
using Fact~\ref{fact:alpha-t-f} and Cassini's identity, and using the fact that the determinant is
multilinear and alternating,
we obtain (cf.\ Appendix~\ref{apx:jacobian}):
\begin{restatable}{fact}{jacobian}
  \label{fct:jacobian}
We have:
  $\det(\mJ_{\xi_N}) = 2^{-4N} \times \det(\mJ_\xi)$ .
\end{restatable}

Hence, $\det(\mJ_{\xi_N})$ is not the null polynomial and, as in Section~\ref{sec:sub6}, we can use
Proposition~\ref{prp:85}
to complete the proof of
Proposition~\ref{prp:same-length} (cf. Appendix~\ref{apx:85}).

\section{Proof for Arbitrary Subdivisions}
\label{sec:general}

In this section we finally prove our main result (Result~\ref{res:main}), which
we re-state here:

\begin{theorem}
  \label{thm:main}
  Let $\mathcal{G}$ be an arbitrary family of graphs which is
treewidth-constructible. Then $\probmatch(\mathcal{G})$ is \#P-hard under ZPP
reductions.
\end{theorem}
We will reduce from the problem of counting matchings in~$3$-regular planar graphs of, which is \#P-hard\footnote{Note that, in holographic literature, graphs may be multigraphs (i.e., can have
multiple edges between two nodes) --- see
\cite{multigraphs}. However, inspecting the proof
of~\cite{xia2007computational}, we see that the
graphs are in fact simple.} by~\cite{xia2007computational}.
Our reduction will be similar to that of Section \ref{sec:same-length}, with
the major issue that
the various edges of the input graph can now be subdivided to different lengths.

The proof consists of five steps. In step 1, we show a general
result allowing us to 
assign probabilities to a path of length 4 so as to ``emulate'' the behavior of
any long path of
\emph{even} length. We then revisit the proof of the previous section.
Step 2 extracts the input graph $H$ 
from the treewidth-constructible family.
Step 3 relates the number of matchings of~$H$ to
cardinalities similar to those of the previous section, but taking the
parities of the subdivisions into account. Step 4 then explains how to conclude
using emulation. Last, step~5 works around the issue that the probabilities 
of Step 1 
could be
irrational, by explaining how we can
conclude with sufficiently precise approximations. We now detail these steps.

\subparagraph*{Step 1: Emulating long even paths.}
We start by presenting the main technical tool, namely, how to emulate long paths of even length by paths of length~$4$.

\begin{restatable}[Emulation result]{proposition}{fibo}
  \label{prp:fibo}
 There exist closed-form expressions, denoted
  $\pp(i),\allowbreak\pq(i),\allowbreak\pr(i),\allowbreak\ps(i)$,
  such that for every even integer $i\geq 4$ the
following hold:
\begin{description}
\item[(A)\label{itema}] the expressions evaluate to well-defined probability
  values, i.e., we have\linebreak $0 \leq \pp(i),\pq(i),\pr(i),\ps(i) \leq 1$; and
\item[(B)\label{itemb}] the path of length 4 with probabilities
  $\pp(i),\pq(i),\pr(i),\ps(i)$ behaves like a path of length $i$
    with probabilities $1/2$, i.e., $\Pi^{bb'}_{4}(\pp(i), \pq(i), \pr(i), \ps(i)) \ = \
\Pi^{bb'}_{i}(1/2, \ldots, 1/2)$ for all $b,b' \in \{0,1\}$.
\end{description}
  Further, each of these expressions is of the form $\frac{P \pm \sqrt{Q}}{R}$
  where $P, Q, R$ are polynomials in the Fibonacci numbers $f_{i-1}$ and
  $f_{i-2}$ and in $2^{-i}$, with rational coefficients.
\end{restatable}

\begin{proof}[Proof sketch]
  The result is simple to state, but we did not find an elegant way
  to show it.
  Our proof consists of four steps:
(i) rewriting condition~\ref{itemb} into a simpler equivalent system of
  equations (using Lemma~\ref{lem:fiboexpr}),
  (ii) proving that any solution of that system must be
in~$(0,1)^4$, (iii) exhibiting closed-form expressions that
satisfy the system, found with the help of SageMath;
and (iv) verifying that these expressions are well-defined. 
See Appendix~\ref{apx:fibo}.
\end{proof}

\begin{remark}
\label{rmk:asym}
As~$\Pi^{bb'}_{i}(1/2, \ldots, 1/2)$ is symmetric, one would expect the
  closed-form expressions to satisfy
  $\pp(i) = \ps(i)$ and~$\pq(i) = \pr(i)$.
  However, surprisingly, numerical evaluation (already for $i=6$) shows that our solution does not
  have this property.
\end{remark}

\begin{remark}
\label{rmk:parity}
  It is necessary to require that $i$ is even, as otherwise
  Proposition~\ref{prp:fibo} demonstrably does not hold.  In fact, we can prove that, more generally, the behavior of a
probabilistic path inherently depends on the parity of its length (cf.\
Appendix~\ref{apx:parity-incompatibility}). This is
why we will distinguish even-length and odd-length subdivisions in the sequel.
\end{remark}
\subparagraph*{Step 2: Choosing the graph in~$\bm{\calG}$.}
Let~$H=(V,E)$ be the input to the reduction, i.e., the~$3$-regular planar graph for which we want to
compute $\shpmatch(H)$, and let~$m\colonequals |E|$.  
We first
build the graph~$H_{10} =\sub(H,10)$, writing $H_{10} = (V_{10},E_{10})$ and
we compute~$k \colonequals
|V_{10}|^c$ where~$c$ is the constant from Theorem~\ref{thm:top-min}. Notice
that $H_{10}$ is a planar graph of maximum degree~$3$, and that the size of~$k$ in unary is
polynomial in (the encoding size of) $H$.
Intuitively, this initial subdivision in~$10$ will ensure that we have enough room for our probabilistic gadgets.
Now, we use the
treewidth-constructibility of~$\mathcal{G}$ to build in polynomial time a
graph~$G = (V_G,E_G) \in \mathcal{G}$ such that~$\tw(G) \geq k$, and using
Theorem~\ref{thm:top-min} we compute in ZPP a 
subgraph~$G'$ of~$G$
with a subdivision $\eta_{10}:E_{10} \to \mathbb{N}^+$ of~$H_{10}$ and an isomorphism
from~$\sub(H_{10},\eta_{10})$ to $G'$. 
This gives us
a subdivision $\eta:E
\to \mathbb{N}^+$ of~$H$ and an isomorphism~$f$ from $\sub(H,\eta)$ to $G'$,
with the initial subdivision ensuring
that~$\eta(e) \geq 10$ for each~$e\in E$.  

\subparagraph*{Step 3: Defining the new sets~$S_{\tau,\tau'}$ and linking them to matchings.}
As before, fix an 
orientation~$\overrightarrow{H}$ of~$H$. 
We call an edge $e$ of $H$ \emph{even} if $\eta(e)$ is even, and \emph{odd}
otherwise.
For
$\tau,\tau' \in [m+1]^4$, both indexed in binary, we define~$S_{\tau,\tau'}$ to be the
set of selection functions~$\mu$ of~$H$ such that, for $b,b' \in \{0, 1\}$,
precisely~$\tau_{bb'}$ even edges~$e$ of~$H$
have type~$bb'$ w.r.t.~$\mu$,
and 
precisely~$\tau'_{bb'}$ odd edges~$e$ of~$H$
have type~$bb'$ w.r.t.~$\mu$.
Then, as
in Section~\ref{sec:sub6}, we have:
\begin{equation}
\label{eqn:param-8}
\shpmatch(H) = \sum_{\substack{\tau,\tau' \in [m+1]^4\\ \tau_{01} = \tau_{10} = \tau'_{01} = \tau'_{10} = 0}} |S_{\tau,\tau'}|.
\end{equation}

\subparagraph*{Step 4: Describing the probabilistic graphs and obtaining the system.}
To complete the definition of the reduction, let us build the~$(m+1)^8$
probabilistic graphs on which we 
want to 
invoke the oracle, denoted 
$G(\kappa,\kappa')$ for $\kappa,\kappa' \in [m+1]^4$. 
Let $K \colonequals \max_{\substack{e\in E\\\eta(e) \text{ is even}}}(\eta(e))$
and $K' \colonequals \max_{\substack{e\in E\\\eta(e) \text{ is odd}}}(\eta(e))$
and~$N\colonequals K-6$ and~$N'\colonequals K'-6$.
The underlying graph of $G(\kappa,\kappa')$ is~$G$, every edge~$e\in E_G$ that is not in~$G'$ is assigned
probability zero, and we explain next what is the probability associated to the
edges that are in~$G'$.  Consider~$2\times (m+1)^4$
4-tuples
of probability values $\rho_{\kappa} =
(\rho_{\kappa,00},\rho_{\kappa,01},\rho_{\kappa,10},\rho_{\kappa,11})$ and
$\rho'_{\kappa} = (\rho'_{\kappa',00},\rho'_{\kappa',01},\rho'_{\kappa',10},\rho'_{\kappa',11})$ for $\kappa,\kappa' \in [m+1]^4$,
to be chosen later.  For every
directed edge~$(x,y)\in \overrightarrow{H}$, let~$\gamma \colonequals
\eta(\{x,y\})$ be the length to which it is subdivided in~$G'$. Letting~$f(x), v_1, \dots, v_{\gamma -1}, f(y)$ be the
corresponding path in~$G'$, we set the probabilities of the~$\gamma$
edges along that path as follows:
\begin{itemize}
  \item If~$\gamma$ is even (illustrated in Figure~\ref{fig:emul}):
    \begin{itemize}
      \item $1/2,\rho_{\kappa,00},\rho_{\kappa,01},\rho_{\kappa,10},\rho_{\kappa,11}$ for the
first~$5$ edges,
\item 
    $p(N-\gamma+10),q(N-\gamma+10),r(N-\gamma+10),s(N-\gamma+10)$ for
the next four edges,
        \item $1/2$ for the remaining~$\gamma - 9$ edges.
    \end{itemize}
  \item If~$\gamma$ is odd:
    \begin{itemize}
      \item $1/2,\rho'_{\kappa',00},\rho'_{\kappa',01},\rho'_{\kappa',10},\rho'_{\kappa',11}$ for the
first~$5$ edges,
        \item $p(N'-\gamma+10),q(N'-\gamma+10),r(N'-\gamma+10),s(N'-\gamma+10)$ for
the next four edges,
        \item $1/2$ for the remaining~$\gamma - 9$ edges.
    \end{itemize}
\end{itemize}
We know that $N-\gamma+10$
(resp., $N'-\gamma+10$) is an even
integer when~$\gamma$ is even
(resp., when~$\gamma$ is odd); and it is $\geq 4$ by definition of~$K$ (resp.,
of~$K'$).
Thus, using Proposition~\ref{prp:fibo}
and then Lemma~\ref{lem:concat}, we know that the path that we defined behaves
exactly like the path
$P_{K}(1/2,\rho_{\kappa,00},\rho_{\kappa,01},\rho_{\kappa,10},\rho_{\kappa,11},1/2,\ldots,1/2)$
if~$\gamma$ is even, and
exactly like the path
$P_{K'}(1/2,\rho'_{\kappa',00},\rho'_{\kappa',01},\rho'_{\kappa',10},\rho'_{\kappa',11},1/2,\ldots,1/2)$ if~$\gamma$ is
odd (see again Figure~\ref{fig:emul}).
\begin{figure}
\begin{tikzpicture}[mynode/.style={circle,fill=black, inner sep=2pt}]
\node[mynode,label=above:{$f(x)$}] (x) at (0,0) {};
\node[mynode,label=above:{$f(y)$}] (y) at (13,0) {};
\node[mynode] (v1) at (1,0) {};
\node[mynode] (v2) at (2,0) {};
\node[mynode] (v3) at (3,0) {};
\node[mynode] (v4) at (4,0) {};
\node[mynode] (v5) at (5,0) {};
\node[mynode] (v6) at (6,0) {};
\node[mynode] (v7) at (7,0) {};
\node[mynode] (v8) at (8,0) {};
\node[mynode] (v9) at (9,0) {};
\node[mynode] (v12) at (12,0) {};
\draw (x) -- (v1) node[above,midway] {$\frac{1}{2}$}; 
\draw (v1) -- (v2) node[above,midway] {$\rho_{\kappa,00}$}; 
\draw (v2) -- (v3) node[above,midway] {$\rho_{\kappa,01}$}; 
\draw (v3) -- (v4) node[above,midway] {$\rho_{\kappa,10}$}; 
\draw (v4) -- (v5) node[above,midway] {$\rho_{\kappa,11}$}; 
\draw (v5) -- (v6) node[above,midway] {$p(i)$}; 
\draw (v6) -- (v7) node[above,midway] {$q(i)$}; 
\draw (v7) -- (v8) node[above,midway] {$r(i)$}; 
\draw (v8) -- (v9) node[above,midway] {$s(i)$}; 
\draw (v12) -- (y) node[above,midway] {$\frac{1}{2}$}; 
\draw[dotted] (v9) -- (v12) node[above,midway] {$\frac{1}{2}$ on all edges};
\draw [decorate, decoration = {calligraphic brace, mirror, raise = 6pt}, thick] (v9) --  (v12);
\node at (10.5,-.7) {$\gamma - 10$ edges};
(b
\node[mynode] (bx) at (0,-1.8) {};
\node[mynode] (by) at (13,-1.8) {};
\node[mynode] (bv1) at (1,-1.8) {};
\node[mynode] (bv2) at (2,-1.8) {};
\node[mynode] (bv3) at (3,-1.8) {};
\node[mynode] (bv4) at (4,-1.8) {};
\node[mynode] (bv5) at (5,-1.8) {};
\node[mynode] (bv12) at (12,-1.8) {};
\draw (bx) -- (bv1) node[above,midway] {$\frac{1}{2}$}; 
\draw (bv1) -- (bv2) node[above,midway] {$\rho_{\kappa,00}$}; 
\draw (bv2) -- (bv3) node[above,midway] {$\rho_{\kappa,01}$}; 
\draw (bv3) -- (bv4) node[above,midway] {$\rho_{\kappa,10}$}; 
\draw (bv4) -- (bv5) node[above,midway] {$\rho_{\kappa,11}$}; 
\draw[dotted] (bv5) -- (bv12); 
\draw (bv12) -- (by) node[above,midway] {$\frac{1}{2}$}; 
\draw[dotted] (bv5) -- (bv12) node[above,midway] {$\frac{1}{2}$ on all edges};
\draw [decorate, decoration = {calligraphic brace, mirror, raise = 6pt}, thick] (bv5) --  (bv12);
\node at (8.5,-2.5) {$N$ edges};
\end{tikzpicture}
\caption{The upper path depicts how we set the probabilities along a path $f(x),
  v_1,
\dots, v_{\gamma -1}, f(y)$ corresponding to an edge~$(x,y)\in
\overrightarrow{H}$ such that~$\gamma \colonequals \eta(\{x,y\})$ is even. We
write~$i \colonequals N - \gamma +10$. By Lemma~\ref{lem:concat} and Proposition~\ref{prp:fibo}, this path
has exactly the same behavior as the lower path.}
\label{fig:emul}
\end{figure}
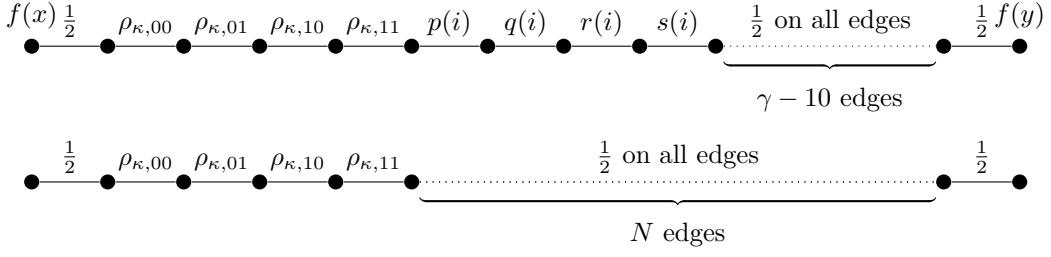

We have now managed to ensure that all paths for even edges (resp., for odd
edges) behave as if they had been subdivided to length~$K$ (resp., to
length~$K'$). We continue the proof as in the previous section, except that we
distinguish odd and even edges. 
Specifically, for~$b,b'\in \{0,1\}$, 
 we write as in the previous section
$\Upsilon_{\kappa,bb'} \colonequals
\Pi^{bb'}_{K-2}(\rho_{\kappa},1/2,\ldots,1/2)$ and
$\Upsilon'_{\kappa',bb'} \colonequals
\Pi^{bb'}_{K'-2}(\rho'_{\kappa'},1/2,\ldots,1/2)$.
Using the same reasoning as for Equation~\ref{eqn:match-to-prob-same}, we 
obtain:
\begin{align}
\label{eqn:match-to-prob-general}
\!\!\!\!\!2^{2m} \times \!\!\!\prm\!\!\!\!(G(\kappa,\kappa')) \ = \ &
  \!\!\!\!\!\!\!\sum_{\tau,\tau' \in [m+1]^4}\!\!\!\! |S_{\tau,\tau'}|
 \times (\Upsilon_{\kappa,00})^{\tau_{00}} \times (\Upsilon_{\kappa,01})^{\tau_{01}} \times (\Upsilon_{\kappa,10})^{\tau_{10}}\times (\Upsilon_{\kappa,11})^{\tau_{11}}
\nonumber
\\
& \qquad\times (\Upsilon'_{\kappa',00})^{\tau'_{00}} \times (\Upsilon'_{\kappa',01})^{\tau'_{01}} \times (\Upsilon'_{\kappa',10})^{\tau'_{10}} \times (\Upsilon'_{\kappa',11})^{\tau'_{11}},
\end{align}
i.e., we obtain a system of linear equations $\mG \mS = \mC$
with $\mS$ the vector of the desired values $|S_{\tau,\tau'}|$,
with $\mC$ the vector of the oracle answers $\prm(G(\kappa,\kappa'))$, and
with $\mG \in
\mathbb{R}^{[m+1]^8, [m+1]^8}$, whose entries are given according to the above
equation.  But notice that we have~$\mG = \mV \otimes \mVP$, with 
$v_{\kappa,\tau} \ \colonequals 2^{-m} \times (\Upsilon_{\kappa,00})^{\tau_{00}} \times (\Upsilon_{\kappa,01})^{\tau_{01}} \times (\Upsilon_{\kappa,10})^{\tau_{10}} \times (\Upsilon_{\kappa,11})^{\tau_{11}} $ and
$v'_{\kappa',\tau'} \ \colonequals 2^{-m} \times (\Upsilon'_{\kappa',00})^{\tau'_{00}} \times (\Upsilon'_{\kappa',01})^{\tau'_{01}} \times (\Upsilon'_{\kappa',10})^{\tau'_{10}} \times (\Upsilon'_{\kappa',11})^{\tau'_{11}} $.
Since~$\mV$ and~$\mVP$ share no variables and are identical up to renaming
variables, to argue that there exist 4-tuples of probabilistic values
$\rho_{\kappa}$ and~$\rho'_{\kappa'}$ for $\kappa,\kappa' \in [m+1]^4$
that make~$\mG$ invertible, it is enough to know that the Jacobian determinant of the
mapping~$\xi_{N}$
is not identically null, as we showed in the previous section (Fact~\ref{fct:jacobian}).
Thus, we can again
use Proposition~\ref{prp:85}
to compute in polynomial
time $2\times (m+1)^4$
$4$-tuples of rational probability
values $\rho_\kappa$ and $\rho'_{\kappa'}$ such that the matrices~$\mV$ and~$\mV'$, hence~$\mG$,
are invertible (cf. Appendix~\ref{apx:85}). By Equation~\ref{eqn:match-to-prob-general}, $\mG$ has rational entries, and its
inverse $\mG^{-1}$ also does and is computable in polynomial time.

\subparagraph*{Step 5: Using decimal fractions approximations.}
The last issue is that we cannot really obtain $\mC$ via
oracle calls, because the graphs $G(\kappa,\kappa')$
may
have irrational edge probabilities, namely, the $p(i),q(i),r(i),s(i)$.
We now argue that we can still recover the $\mC$, so that we can
compute $\mS = \mG^{-1} \mC$ and conclude.
To do this, we first observe 
that~$\mC$ is in fact a vector of decimal fractions, as the graphs
$G(\kappa,\kappa')$ emulate a graph where the probabilities are
decimal fractions; further, we can bound the number of
decimal places of its values to $[m \times (\max(N,N') + 10)] \times z$, with $z$ the maximal
number of decimal places of a decimal fraction in $\rho_\kappa, \rho'_{\kappa}$
Second, we show how to compute 
decimal fraction \emph{approximations} $\widehat{\pp(i)}, \widehat{\pq(i)}, \widehat{\pr(i)}, \widehat{\ps(i)}$ of the $p(i),q(i),r(i),s(i)$,
in polynomial time in the desired number of places, using the form
that they have according to
Proposition~\ref{prp:fibo}.
Third, we argue that when invoking the oracles on the graphs where we replace
$p(i),q(i),r(i),s(i)$ by $\widehat{\pp(i)}, \widehat{\pq(i)}, \widehat{\pr(i)},
\widehat{\ps(i)}$, then the error on the answer is bounded as a function of that of the
approximations, so that we can recover $\mC$ exactly if the approximations
were sufficiently precise.
See Appendix~\ref{apx:precision} for detailed proofs.

\section{Result for Edge Covers}
\label{sec:others}
Having shown Result~\ref{res:main},
we now explain how to adapt its proof to obtain our analogous results for 
edge covers.
We only sketch the argument, and refer to
Appendix~\ref{apx:others} for more details.
Recall that an \emph{edge cover} of a
graph~$G=(V,E)$ is a set of edges~$S \subseteq E$ such that~$V = \bigcup_{e\in S}e$. Given
a probabilistic graph~$(G,\pi)$, we define $\pre(G,\pi)$ to be the sum of the
probabilities of all edge covers in the probability distribution induced
by~$\pi$, and define $\probedgecover(\calF)$ for a graph family~$\calF$ 
to be the corresponding computational problem.
We first note that, in this context, the strict analogue of
Result~\ref{res:main} does not hold. Indeed, take
some treewidth-constructible graph family $\calG$, and consider the graph family~$\calG'$ 
obtained from~$\calG$ as follows: for every graph~$G \in \calG$, we add
to~$\calG'$ the graph that is obtained from~$G$ by 
attaching a dangling edge with a fresh vertex to
every node of~$G$.
The family $\calG'$ is still 
treewidth-constructible, but $\probedgecover(\calG')$ is now tractable
as it is easy to see that the edge covers of a graph in~$\calG'$ are precisely the edge subsets where
all dangling edges are kept.

To avoid this, let us assume that~$\calG$ is closed
under taking subgraphs, i.e., if~$G\in \calG$ and~$G'$ is a
subgraph of~$G$, then~$G' \in \calG$. We then have:

\begin{theorem}
  \label{thm:edgecovers}
  Let $\mathcal{G}$ be an arbitrary family of graphs which is
treewidth-constructible and closed under taking subgraphs. Then $\probedgecover(\mathcal{G})$ is \#P-hard under ZPP
reductions.
\end{theorem}

This is proved like Result~\ref{res:main}, with the following
modifications. We reduce from counting edge covers (instead of matchings) on
$3$-regular planar graphs: this is hard by~\cite{cai2012holographic}, even on
simple graphs~\cite[Appendix~D]{amarilli2017conjunctive}. We now define
a selection function~$\mu$ to map
each vertex~$x \in V$ to \emph{at least} one incident edge, and we define the
types and the sets~$S_{\tau,\tau'}$ as before, via an arbitrary orientation of
the graph~$H$. We obtain the number of edge covers of~$H$ from the
quantities~$|S_{\tau,\tau'}|$ exactly as in Equation~\ref{eqn:param-8}. We
redefine  $\Pi^{bb'}_{n}(\rho)$ to be the probability of an edge cover in a
path of length~$n$ with probabilities~$\rho$ on the edges and with endpoint
constraints given by~$b,b'$ as before. Lemma~\ref{lem:fiboexpr} then becomes
$\Pi^{bb'}_{n}(1/2,\ldots,1/2) = \frac{f_{n+b+b'}}{2^n}$, i.e., the role
of~$b,b'$ is “reversed”. Analogous versions of Lemma~\ref{lem:concat} and of Proposition~~\ref{prp:fibo}
still hold, so the relevant Jacobian determinants are still non-identically null.
We take the graph~$G\in \calG$ again via the topological minor
extraction result, but this time directly extracting
$\sub(H,\eta) \in \calG$
as $\calG$ is subgraph-closed.
The rest of the proof is identical.

We point out that
the situation is different
for \emph{perfect} matchings.  Indeed, using a weighted variant of the FKT
algorithm~\cite[Chapter 4]{cai2017complexity},
the weighted counting of perfect matchings is polynomial-time over the class of planar graphs,
which is treewidth-constructible.

We conclude by leaving open two directions for future work. The first one would be to obtain
the same kind of lower bounds when the
probabilities annotate the \emph{nodes} instead of the edges, that is, studying
the corresponding weighted counting problems for, e.g., independent sets, vertex covers, or cliques. 
We believe that the corresponding result should hold and do not expect any
surprises.
The second question would be to show our hardness results in the \emph{unweighted} case,
e.g., unweighted counting of matchings, assuming that the graph family is
subgraph-closed. This appears to be much more challenging, as
our current proof crucially relies on the ability to use arbitrary probability values.

\bibliography{main.bib}

\appendix
\section{Hardness under ZPP reductions}
\label{apx:formal}
We formally show the result about hardness under ZPP reduction claimed in the
preliminaries:

\begin{claim}
  If $P_1$ is a problem which is \#P-hard, and $P_1$ has a
  ZPP-reduction to a problem~$P_2$, then $P_2$ is \#P-hard under ZPP reductions.
\end{claim}

\begin{proof}
  Let $P_0$ be a \#P-hard problem, and let us show that $P_0$ has a
  ZPP-reduction to~$P_2$. We will do so by composing the two reductions. The only
  subtle point is that the reductions are Turing reductions, so we must control
  the overall probability of failure knowing that each reduction to the
  $P_2$-oracle may fail. Specifically, given an instance $I$ to~$P_0$, perform the
  polynomial-time reduction using an oracle to~$P_1$, and evaluate the calls
  to~$P_1$ by performing the ZPP-reduction to~$P_2$ and using the oracle to~$P_2$.
  Let $n$ be the total number of calls of the ZPP-reduction to~$P_2$; it is a
  polynomial in~$|I|$. Let $c$ be the failure probability of the ZPP-reduction,
  and let $k$ be sufficiently large so that $n c^k \leq c$; this is polynomial
  in~$n$. When performing the ZPP-reduction to~$P_2$, repeat each call $k$ times.
  This guarantees that, by the union bound, the total probability of failure is
  at most $n c^k \leq c$, i.e., it is at most a constant; so indeed we have
  defined a ZPP reduction from~$P_0$ to~$P_2$.
\end{proof}

Following this claim, to prove Result~\ref{res:main}, we will establish
\#P-hardness under ZPP reductions simply by giving a ZPP reduction from a
\#P-hard problem.

\section{Proof of Fact~\ref{fct:param}}
\label{apx:param}
We prove Fact~\ref{fct:param}, whose statement we recall:
\fctparam*
Indeed, let us show that there is a bijection between the matchings of~$H$ and
the set $\bigcup_{\substack{\tau \in [m+1]^4 \\ \tau_{01} = \tau_{10} = 0}} S_{\tau}$; as
the~$S_\tau$ are pairwise disjoint, this is enough to prove the
claim. If~$M$ is a matching of~$H$, let~$\mu_M$ be the selection function
of~$H$ that assigns~$\emptyset$ to every node that is not in any edge
of~$M$, and that assigns~$\{e\}$ to every node that is in~$e$ for some 
edge~$e\in M$.  It is easy to see that~$\mu_M$ is a well-defined selection
function (because~$S$ is a matching), that it is in
some~$S_{\tau}$ with~$\tau_{01} = \tau_{10} = 0$, and that the mapping~$M\mapsto
\mu_M$ is injective. Furthermore, any selection function~$\mu$ that is in
some~$S_{\tau}$ with~$\tau_{01} = \tau_{10} = 0$ can be obtained as~$\mu_M$ for some matching~$M$ of~$H$: take~$M\colonequals \{\{x,y\}\mid
\{x,y\} \in \mu(x) \cap \mu(y)\}$, i.e., $\tau_{11}$ is the cardinality of the matching
and $\tau_{00} = m-\tau_{11}$ is the cardinality of the complement of the matching.
Hence we indeed have a bijection between the matchings of~$H$ and
$\bigcup_{\substack{\tau \in [m+1]^4 \\ \tau_{01} = \tau_{10} = 0}} S_{\tau}$.

\section{Proof of Fact~\ref{fact:match-to-prob}}
\label{apx:match-to-prob}
We prove Fact~\ref{fact:match-to-prob}, whose statement we recall here:

\matchtoprob*

Indeed, for a matching~$M_6$ of~$H_6$, consider the selection
function~$\mu_{M_6}$ defined as follows: for every directed edge~$(x,y)$
of~$\overrightarrow{H}$, letting $x \edgeg v_1 \edgeg v_2 \edgeg v_3 \edgeg v_4 \edgeg v_5 \edgeg y$ be the
corresponding path in~$H_6$, we add~$\{x,y\}$ to~$\mu_{M_6}(x)$ if the first
edge of that path is in~$M_6$, and we add~$\{x,y\}$ to~$\mu_{M_6}(y)$ if the
last edge of that path is in~$M_6$. Observe that~$\mu_{M_6}$ is indeed a
selection function, and that it is in~$S_\tau$ for exactly one
$4$-tuple~$\tau \in [m+1]^4$. We now use the definition of $\prm(H_6(\kappa))$, i.e., Equation~\eqref{eqn:match}, and split the sum according to the
$4$-tuple~$\tau$ for which the selection function~$\mu_{M_6}$ is
in~$S_\tau$, and then split again according to the 
specific~$\mu \in S_\tau$
to which the selection function is equal, as
follows.
\begin{align}
\nonumber \prm(H_6(\kappa)) \ &= \ \sum_{\text{matching }M_6\text{ of }H_6} \Pr_{H_6(\kappa)}(M_6)\\
\label{eqn:bla}&= \ \sum_{\tau \in [m+1]^4} \ \ \ \sum_{\mu \in S_\tau} \  \sum_{\substack{\text{matching }M_6\text{ of }H_6\\ \text{s.t.\ } \mu_{M_6} = \mu}} \ \ \ \Pr_{H_6(\kappa)}(M_6).
\end{align}
But, by definition of~$H_6(\kappa)$
and of the~$\Lambda_{\kappa,bb'}$, and because the edges are independent, for a selection function~$\mu \in S_\tau$ we have

\begin{align*}
\sum_{\substack{\text{matching }M_6\text{ of}\\ H_6\text{ s.t.\ } \mu_{M_6} = \mu}} \Pr_{H_6(\kappa)}(M_6) &=\ \frac{1}{2^{2m}} \times \left(\prod_{\substack{\{x,y\} \in E\\ \text{s.t.\ }\{x,y\} \notin \mu(x)\cup \mu(y)}} \Lambda_{\kappa,00}\right) \left(\prod_{\substack{\{x,y\} \in E\\ \text{s.t.\ }\{x,y\} \in \mu(y)\setminus \mu(x)}} \Lambda_{\kappa,01}\right)\\
&\phantom{lalalala} \times \left(\prod_{\substack{\{x,y\} \in E\\ \text{s.t.\ }\{x,y\} \in \mu(x)\setminus \mu(y)}} \Lambda_{\kappa,10}\right) \left(\prod_{\substack{\{x,y\} \in E\\ \text{s.t.\ }\{x,y\} \in \mu(x)\cap \mu(y)}} \Lambda_{\kappa,11}\right)\\
&= \  \frac{1}{2^{2m}}\times (\Lambda_{\kappa,00})^{\tau_{00}} \times (\Lambda_{\kappa,01})^{\tau_{01}} \times (\Lambda_{\kappa,10})^{\tau_{10}} \times (\Lambda_{\kappa,11})^{\tau_{11}}.
\end{align*}
Injecting the above into Equation~\eqref{eqn:bla},
and noticing that it only depends on $\tau$ and $\kappa$,
we can factorize the resulting expression to obtain Fact~\ref{fact:match-to-prob}, as wanted.

\section{Finding in PTIME Rational Values that Make the Matrix Invertible}
\label{apx:85}
In this section we give more details about the result by Dalvi and Suciu that allows us
throughout the paper to find rational values of the probabilities~$\rho_\kappa$
that make the relevant matrices from
Sections~\ref{sec:sub6},~\ref{sec:same-length} and~\ref{sec:general} have
rational entries and be invertible.  
This result appears as Proposition 8.44 of~\cite{dalvi2012dichotomy}
and is restated as Proposition~\ref{prp:85} in our paper.

The original result of Dalvi and Suciu only mentions \emph{rational}
probability values instead of decimal fractions, but this would not be good
enough for our proof in Section~\ref{sec:general}.  Fortunately, an inspection
of the proof of this result reveals that this still holds if we want to compute
decimal fractions: indeed, the proof relies on the 
multivariate version of the following observation to find the rational values: if~$f$ is a polynomial of
degree~$d$ in one variable that is not the null polynomial, to find a value
that is not a root of~$f$ it is enough to try~$d+1$ distinct values of the
variable. This obviously still works if we additionally require that the values
are decimal fractions (see \cite[Section 8.5]{dalvi2012dichotomy} for the full
proof).

In our context, we only use this result with $k = 4$ and~$I=J=[m+1]^4$. For
Section~\ref{sec:sub6}, the polynomials~$P_\tau$ are those mentioned in the
“Making $\mV$ invertible” paragraph.  For Section~\ref{sec:same-length} the
polynomials~$P_\tau$ are this time given according to
Equation~\ref{eqn:match-to-prob-same}.  For Section~\ref{sec:general}, we use
it independently
for 
the two matrices~$\mV$ and~$\mV'$, which are both identical (up to renaming of variables) to the matrix from
Section~\ref{sec:same-length}.

\section{Explicit Computation of the Jacobian Determinant}
\label{apx:jacobian-first}
By explicit computation with the help of SageMath (see
\verb|jacobian-pqrs.ipynb| in supplementary material~\cite{supmat}) we obtain that 
$\det(\mJ_\xi) =
\chi_{00}\chi_{11}(\chi_{00}+\chi_{10}-\chi_{00}\chi_{10})(1-\chi_{00})(1-\chi_{01})^2(1-\chi_{11})^3$.  
This evaluates to $\frac{1}{128}$
for~$\chi_{00}=\chi_{11}=1/2$, $\chi_{01}=\chi_{10}=0$.

\section{Proof of Lemma~\ref{lem:concat}}
\label{apx:concat}
We prove Lemma~\ref{lem:concat}, whose statement we recall here.

\concat*

The intuition for the result is that the edge subsets that are matchings are
those where the $(n+1)$-th edge is not kept, plus those where the $n$-th edge is
not kept, minus those where the $n$-th and $(n+1)$-th edge were not kept (as
these were counted twice).

We now give the formal proof.
We prove the claim for~$b=b'=0$, the other cases being similar. Let us
see~$P_{n+n'}$ as the concatenation of the two paths~$P_n$ and~$P_{n'}$,
where~$P_n$ and~$P_{n'}$ are disjoint except for the connecting node, and
define the following sets and quantities, for~$b,b'\in \{0,1\}$:

\begin{itemize}
  \item $M_{n+n'}^{bb'}$ is the set of matchings~$M$ of~$P_{n+n'}$ such that
the~$n$-th edge is in~$M$ iff~$b=1$ and the~$(n+1)$-th edge is in~$M$
iff~$b'=1$ (in particular $M_{n+n'}^{11}$ is empty); and~$\alpha^{bb'}$ is
the probability of $M_{n+n'}^{bb'}$ in $P_{n+n'}(\rho,\rho')$.
  \item $M_{n}^{\bullet, b}$ is the set of matchings~$M$ of~$P_n$ such that, if~$b=1$
then the~$n$-th edge is not in~$M$; in particular,
$\Pi_n^{0,b}(\rho)$ is the probability of~$M_{n}^{\bullet, b}$ in
$P_{n}(\rho)$
  \item $M_{n'}^{b, \bullet}$ is the set of matchings~$M$ of~$P_{n'}$ such that,
if~$b=1$ then the first edge is not in~$M$; in particular,
$\Pi_{n'}^{b,0}(\rho')$ is the probability of~$M_{n'}^{b,
    \bullet}$ in $P_{n'}(\rho')$.\\
\end{itemize}
Observe that the set of matchings of~$P_{n+n'}$ is the disjoint union
of~$M_{n+n'}^{00}$, $M_{n+n'}^{01}$ and $M_{n+n'}^{10}$, so that we have
\begin{equation}
\label{eqn:stup}
\Pi_{n+n'}^{00}(\rho,\rho') = \alpha^{00} + \alpha^{01} + \alpha^{10}.\end{equation}
Now, if~$S$ is a set of sets of edges of~$P_n$ and~$S'$ a set of sets of edges of~$P_{n'}$, define
\linebreak$S \odot S' \colonequals \{M \cup M'  \mid M\in S \text{ and } M' \in S'\}$.
Then, observe that

\begin{itemize}
  \item $M_n^{\bullet,0} \odot M_{n'}^{1,\bullet} = M_{n+n'}^{00} \cup M_{n+n'}^{10}$; and
  \item $M_n^{\bullet,1} \odot M_{n'}^{0,\bullet} = M_{n+n'}^{00} \cup M_{n+n'}^{01}$; and
  \item $M_n^{\bullet,1} \odot M_{n'}^{1,\bullet} = M_{n+n'}^{00}$.\\
\end{itemize}
The above right-hand side unions being disjoint, this implies that ($\star$):

\begin{itemize}
  \item $\Pr(M_n^{\bullet,0} \odot M_{n'}^{1,\bullet}) = \alpha^{00} + \alpha^{10}$; and
  \item $\Pr(M_n^{\bullet,1} \odot M_{n'}^{0,\bullet}) = \alpha^{00} + \alpha^{01}$; and
  \item $\Pr(M_n^{\bullet,1} \odot M_{n'}^{1,\bullet}) = \alpha^{00}$,\\
\end{itemize}
where the probability distribution of~$\Pr$ is that of
$P_{n+n'}(\rho,\rho')$. Now, using the fact that
the edges are independent, notice that~$\Pr(M_n^{\bullet,b} \odot
M_{n'}^{b',\bullet}) = \Pi_n^{0,b}(\rho) \times \Pi_{n'}^{b',0}(\rho')$.
Combining this observation together with ($\star$) and Equation~\eqref{eqn:stup}
concludes.

\section{Proof of Lemma~\ref{lem:fiboexpr}}
\label{apx:fiboexpr}
Here we prove Lemma~\ref{lem:fiboexpr}, whose statement we recall.

\fiboexpr*

We first show the claim for $b = b' = 0$ by induction on $n$. Recall that we
have $\Pi^{00}_n(1/2,\ldots,1/2) =  \Pi_n(1/2,\ldots,1/2)$. We have 
$\Pi_1(1/2) = 1$ and $f_3/2^1 = 1$, and $\Pi_2(1/2) = \frac{3}{4}$ and $f_4/2^2 = \frac{3}{4}$.  For $n>2$, we have:
\begin{align*}
\Pi^{00}_n(1/2,\ldots,1/2) &= \Pi_n(1/2,\ldots,1/2) = \frac{1}{2} \Pi_{n-1}(1/2,\ldots,1/2) + \frac{1}{2} \frac{1}{2} \Pi_{n-2}(1/2,\ldots,1/2)\\
& = \frac{1}{2} \Pi^{00}_{n-1}(1/2,\ldots,1/2) + \frac{1}{2} \frac{1}{2} \Pi^{00}_{n-2}(1/2,\ldots,1/2).
\end{align*}

Using the induction hypothesis, we have:
\[
\Pi^{00}_n(1/2,\ldots,1/2) = \frac{1}{2} \frac{f_{n+1}}{2^{n-1}} + \frac{1}{2} \frac{1}{2}
\frac{f_{n}}{2^{n-2}}.
\]
Hence:
\[
\Pi^{00}_n(1/2,\ldots,1/2) = \frac{1}{2^n} (f_{n+1} + f_{n}).
\]
Now, the definition of the Fibonacci sequence concludes.

Second, we show the claim for arbitrary $b, b'$. It is clear that
$\Pi^{01}_{n}(1/2,\ldots,1/2) = \frac{1}{2} \Pi^{00}_{n-1}(1/2,\ldots,1/2)$
because the last edge must be absent, and then the condition on the~$n-1$
remaining edges is the same as the condition on $P_{n-1}(1/2, \ldots, 1/2)$.
For the same reason, $\Pi^{10}_{n}(1/2,\ldots,1/2) = \frac{1}{2}
\Pi^{00}_{n-1}(1/2,\ldots,1/2)$, and $\Pi^{11}_{n}(1/2,\ldots,1/2) =
\frac{1}{2^2} \Pi^{00}_{n-2}(1/2,\ldots,1/2)$. Hence, we obtain the claimed
equalities.

\section{Proof of Fact~\ref{fct:jacobian}}
\label{apx:jacobian}
We prove Fact~\ref{fct:jacobian}, whose statement we recall:
\jacobian*

Let us first consider the Jacobian $\mJ_{\xi_N}$, which is as follows:

           \begin{equation*}
           \mJ_{\xi_N} = 
           \begin{pmatrix}
             \frac{\partial \Upsilon_{\chi,00}}{\partial \chi_{00}} & \frac{\partial \Upsilon_{\chi,00}}{\partial \chi_{01}} & \frac{\partial \Upsilon_{\chi,00}}{\partial \chi_{10}} & \frac{\partial \Upsilon_{\chi,00}}{\partial \chi_{11}} \\
             \frac{\partial \Upsilon_{\chi,01}}{\partial \chi_{00}} & \frac{\partial \Upsilon_{\chi,01}}{\partial \chi_{01}} & \frac{\partial \Upsilon_{\chi,01}}{\partial \chi_{10}} & \frac{\partial \Upsilon_{\chi,01}}{\partial \chi_{11}} \\
             \frac{\partial \Upsilon_{\chi,10}}{\partial \chi_{00}} & \frac{\partial \Upsilon_{\chi,10}}{\partial \chi_{01}} & \frac{\partial \Upsilon_{\chi,10}}{\partial \chi_{10}} & \frac{\partial \Upsilon_{\chi,10}}{\partial \chi_{11}} \\
             \frac{\partial \Upsilon_{\chi,11}}{\partial \chi_{00}} & \frac{\partial \Upsilon_{\chi,11}}{\partial \chi_{01}} & \frac{\partial \Upsilon_{\chi,11}}{\partial \chi_{10}} & \frac{\partial \Upsilon_{\chi,11}}{\partial \chi_{11}} \\
           \end{pmatrix}.
           \end{equation*}

And recall from Section~\ref{sec:sub6} that we have:
           \begin{equation*}
           \mJ_\xi = 
           \begin{pmatrix}
             \frac{\partial \Lambda_{\chi,00}}{\partial \chi_{00}} & \frac{\partial \Lambda_{\chi,00}}{\partial \chi_{01}} & \frac{\partial \Lambda_{\chi,00}}{\partial \chi_{10}} & \frac{\partial \Lambda_{\chi,00}}{\partial \chi_{11}} \\
             \frac{\partial \Lambda_{\chi,01}}{\partial \chi_{00}} & \frac{\partial \Lambda_{\chi,01}}{\partial \chi_{01}} & \frac{\partial \Lambda_{\chi,01}}{\partial \chi_{10}} & \frac{\partial \Lambda_{\chi,01}}{\partial \chi_{11}} \\
             \frac{\partial \Lambda_{\chi,10}}{\partial \chi_{00}} & \frac{\partial \Lambda_{\chi,10}}{\partial \chi_{01}} & \frac{\partial \Lambda_{\chi,10}}{\partial \chi_{10}} & \frac{\partial \Lambda_{\chi,10}}{\partial \chi_{11}} \\
             \frac{\partial \Lambda_{\chi,11}}{\partial \chi_{00}} & \frac{\partial \Lambda_{\chi,11}}{\partial \chi_{01}} & \frac{\partial \Lambda_{\chi,11}}{\partial \chi_{10}} & \frac{\partial \Lambda_{\chi,11}}{\partial \chi_{11}} \\
           \end{pmatrix}.
           \end{equation*}

Let us call $L_{00}$, $L_{01}$, $L_{10}$, $L_{11}$ the lines of
the above, so that we write it:
\begin{equation*}
\mJ_{\xi} = 
\begin{pmatrix}
L_{00} \\
L_{01} \\
L_{10} \\
L_{11} \\
\end{pmatrix}.
\end{equation*}

Now, using Fact~\ref{fact:alpha-t-f}, we can express the lines of the matrix
$\mJ_{\xi_N}$ as linear combinations of the $L_{b,b'}$, thus:

\begin{equation*}
\mJ_{\xi_N} = \frac{1}{2^{4N}}
\begin{pmatrix}
L_{00} f_{N+1} + L_{01} f_N \\
L_{00} f_N + L_{01} f_{N-1} \\
L_{10} f_{N+1} + L_{11} f_N \\
L_{10} f_N + L_{11} f_{N-1} \\
\end{pmatrix}.
\end{equation*}

We now use two properties of the determinant:
\begin{itemize}
\item It is multilinear, so that the determinant of the above matrix can be
expressed as the sum of the determinants of the 16 matrices obtained by choosing
one term in each row
\item It is alternating: the terms where the same $L_{b,b'}$ occurs twice (even
with different coefficients) have a determinant of zero
\end{itemize}

Thus, we can write $\det(\mJ_{\xi_N})$ as:
\begin{align*}
  \det(\mJ_{\xi_N}) = & \frac{1}{2^{4N}} \left(
f_{N+1}^2 f_{N-1}^2
\det
\begin{pmatrix}
L_{00} \\
L_{01} \\
L_{10} \\
L_{11} \\
\end{pmatrix}
+
f_{N+1} f_{N-1} f_N^2
\det
\begin{pmatrix}
L_{00} \\
L_{01} \\
L_{11} \\
L_{10} \\
\end{pmatrix}\right.
  \\ &
\left.
  \quad\quad\quad\quad
\quad\quad\quad\quad
  +
f_{N+1} f_{N-1} f_N^2
\det
\begin{pmatrix}
L_{01} \\
L_{00} \\
L_{10} \\
L_{11} \\
\end{pmatrix}
+
f_N^4
\det
\begin{pmatrix}
L_{00} \\
L_{01} \\
L_{11} \\
L_{10} \\
\end{pmatrix}
\right)
\end{align*}

Note that the determinant in the first term is 
that of $\mJ_{\xi}$. As for the others, we use the fact that swapping two rows of a matrix multiplies the
determinant by $-1$, to obtain the same determinants. This gives:
\[
\det(\mJ_{\xi_N}) = \frac{
f_{N+1}^2 f_{N-1}^2
-2
f_{N+1} f_{N-1} f_N^2
+
f_N^4
}{2^{4N}} \times \det(\mJ_{\xi}) 
\]

We now observe that:
\[
f_{N+1}^2 f_{N-1}^2 -2 f_{N+1} f_{N-1} f_N^2 + f_N^4 = (f_{N+1} f_{N-1} -
f_N^2)^2
\]

By Cassini's identity, the right-hand side evaluates to 1, so that we have the
claimed equality.

\section{Proof for Remark~\ref{rmk:parity}}
\label{apx:parity-incompatibility}
In this section we prove the following proposition, which justifies that we
need to redefine the parameters~$S_\tau$ to differentiate between even-length
and odd-length subdivisions (cf.\ Remark~\ref{rmk:parity}).  Note that this is
not necessary for the proof of our main result and can safely be skipped; we
only include it for completeness.

\begin{proposition}
\label{prp:parity-incompatibility}
Let~$n\in \mathbb{N}^+$ and~$\rho \in [0,1]^n$, and let us consider
the “determinant-style” quantity \[D = \Pi_n^{01}(\rho) \times
\Pi_n^{10}(\rho) - \Pi_n^{00}(\rho) \times
\Pi_n^{11}(\rho).\]
  Then~$D$ is null if and only if one of the $\rho_i$ is $0$ or~$1$. Otherwise, $D$
  is positive if $n$ is even, and negative if $n$ is odd.
\end{proposition}

The proof is by induction on $n$.
We note that similar reasoning appears in \cite[Lemma
7.2]{amarilli2022uniform}.

  For $n=1$, we have $\Pi_1^{00}(\rho_0) = 1$ and $\Pi_1^{bb'}(\rho_0) =
1-\rho_0$ for $b, b' \in \{0, 1\}$ such that $(b, b') \neq (0, 0)$.  Thus $D =
(1-\rho_0)^2 - 1+\rho_0 = -\rho_0(1-\rho_0)$, which is zero if $\rho_0=0$ or
$\rho_0=1$ and negative otherwise.

  For~$n=2$,  we have:
 \begin{itemize}
   \item $\Pi_2^{00}(\rho_0, \rho_1) = 1 - \rho_0 \rho_1$
   \item $\Pi_2^{10}(\rho_0, \rho_1) = 1 - \rho_0$
   \item $\Pi_2^{01}(\rho_0, \rho_1) = 1 - \rho_1$
   \item $\Pi_2^{11}(\rho_0, \rho_1) = (1 - \rho_0)(1 - \rho_1)$
 \end{itemize}
 Thus $D = (1-\rho_0)(1-\rho_1) - (1-\rho_0\rho_1)(1-\rho_0)(1-\rho_1) = (1-\rho_0)(1-\rho_1) \rho_0 \rho_1$.
 This is zero if one of $\rho_0$ or $\rho_1$ is equal to $0$ or~$1$, and
 positive otherwise.

  We now reason by induction.
  For brevity we omit the arguments of $\Pi$
and use $\Pi_n^{bb'}$ for $\Pi_n^{11}(\rho)$ and $\Pi_{n-2}^{bb'}$ for $\Pi_{n-2}^{bb'}(\rho_1, \ldots,
  \rho_{n-2})$. We also write for brevity $\overline{p} = 1-p$ for~$p\in [0,1]$.
  For $n>2$, we have:
  \begin{itemize}
    \item $\Pi_n^{11} =
      \overline{\rho_0}\,\overline{\rho_{n-1}} \Pi_{n-2}^{00}$
    \item $\Pi_n^{01} =
      \Pi_n^{11} +  \rho_0 \overline{\rho_{n-1}} \Pi_{n-2}^{10}$ because the possible edge subsets with an added edge to the
      right are those with an added edge to both ends, plus those where the
      $\rho_0$ edge is kept (and the $\rho_{n-1}$ edge is not kept)
    \item $\Pi_n^{10} =
      \Pi_n^{11} + \overline{\rho_0} \rho_{n-1} \Pi_{n-2}^{01}$ (symmetrically)
    \item $\Pi_n^{00} =
      \Pi_n^{01} + \Pi_n^{10} -
      \Pi_n^{11} + \rho_0 \rho_{n-1} \Pi_{n-2}^{11}$ because the possible edge subsets with no added edges are
      those where the $\rho_0$ edge must not be kept, plus those where the
      $\rho_{n-1}$ edge must not be kept, minus those where the $\rho_0$ and $\rho_{n-1}$
      edge must both not be kept (double counts), plus those where the $\rho_0$ and
      $\rho_{n-1}$ edges must both be kept. Expanding the definitions of
      $\Pi_n^{01}$ and $\Pi_n^{10}$, this rewrites to:
      $\Pi_n^{00} =
      \Pi_n^{11} + \overline{\rho_{n-1}} \rho_0 \Pi_{n-2}^{10} + \overline{\rho_0} \rho_{n-1} \Pi_{n-2}^{01}
      + \rho_0 \rho_{n-1} \Pi_{n-2}^{11}$
  \end{itemize}
  Let us compute~$D$:
  \begin{align*}
    D = &(\Pi_n^{11} + \overline{\rho_{n-1}}\rho_0\Pi^{10}_{n-2}) 
    (\Pi_n^{11} + \overline{\rho_0}\rho_{n-1}\Pi^{01}_{n-2})\\
    &- \Pi_n^{11} (\Pi_n^{11} + \overline{\rho_{n-1}} \rho_0 \Pi_{n-2}^{10} +
    \overline{\rho_0} \rho_{n-1} \Pi_{n-2}^{01}
      + \rho_0 \rho_{n-1} \Pi_{n-2}^{11})
   \end{align*}
    Many terms simplify, leaving:
  \[
    D = \overline{\rho_{n-1}}\rho_0\Pi^{10}_{n-2} \overline{\rho_0}\rho_{n-1}\Pi^{01}_{n-2}
    - \Pi_n^{11} \rho_0 \rho_{n-1} \Pi_{n-2}^{11}
    \]
    Expanding the definition of $\Pi_n^{11}$ and factoring yields:
    \[
      D = \rho_0 \overline{\rho_0}\rho_{n-1}\overline{\rho_{n-1}} (\Pi^{10}_{n-2} \Pi^{01}_{n-2}
  - \Pi_{n-2}^{00}\Pi_{n-2}^{11})
\]
We recognize the expression of the determinant-style expression for $n-2$. Thus,
  $D$ is zero if that expression is zero, i.e., by induction, one of the $\rho_1,
  \ldots, \rho_{n-2}$ is $0$ or $1$; or if one of $\rho_0, \rho_{n-1}$ is $0$ or
  $1$. Thus $D$ is zero iff one of the $\rho_0, \ldots, \rho_{n-1}$ is $0$ or~$1$.
  Otherwise the sign of~$D$ is that of the expression for $n-2$, so the
  induction hypothesis concludes.

\section{Proof of the Emulation Result (Proposition~\ref{prp:fibo})}
\label{apx:fibo}
In this section we prove Proposition~\ref{prp:fibo}. We recall its statement
for the reader's convenience:

\fibo*

We prove Proposition~\ref{prp:fibo} in the rest of this appendix section.
We fix once and for all the even integer~$i\geq 4$.
To prove the result, we first give some general-purpose inequality lemmas about the Fibonacci
sequence in Appendix~\ref{apx:fibolem} which we use in several places. Then we prove Proposition~\ref{prp:fibo} in four steps, 
corresponding to the following four subsections.

First, in Section~\ref{apx:fibo-equiv} we derive a
system of equations, denoted~\ref{iteme}, that is equivalent to \ref{itemb}.
We then we prove in Section~\ref{apx:fibo-bornes} that any tuple of real
numbers $(p,q,r,s)$ that is a solution to~\ref{iteme} must be in~$(0,1)^4$.  In
Section~\ref{apx:fibo-sat} we use SageMath to help us find symbolic expressions
that satisfy system~\ref{iteme}.  Last, in Section~\ref{apx:fibo-welldef}, we
show that these expressions are indeed well-defined.  Putting it all together
gives us Proposition~\ref{prp:fibo}.

\subsection{Inequality Lemmas on the Fibonacci Sequence}
\label{apx:fibolem}
Let us first prove the general-purpose results on the Fibonacci sequence.
We will use \emph{Binet's
formula}, a closed-form expression for~$f_n$ given by~$f_n = \frac{\varphi^n -
(1-\varphi)^n}{\sqrt{5}}$, where~$\varphi \colonequals \frac{1+\sqrt{5}}{2}$
($\approx 1.61$) is the \emph{golden ratio}.  Looking at the statement of the
next three lemmas, notice that it is clear that these inequalities are true
asymptotically, since~$f_n \sim \frac{\varphi^n}{\sqrt{5}}$ when~$n$ goes to
infinity; however, we need to prove that the inequalities hold when starting
from some specific values.  To this end, we will use the following two
trivialities:

\begin{fact}
\label{fact:2-vs-phi}
For all~$\alpha>0$ and~$n \geq n_\alpha \colonequals \frac{\log\alpha}{\log 2 - \log \varphi}$, we have
 $2^n \geq \alpha \times \varphi^n$.
\end{fact}
\begin{proof}
  We have $n_\alpha (\log 2 - \log \varphi) \geq \log \alpha$, noting that $\log 2
  - \log \varphi > 0$ because $\varphi < 2$. Thus, we have $n \log 2 \geq \log \alpha
  + n \log \varphi$. As the exponential is an increasing function, we obtain the claimed inequality.
\end{proof}
and
\begin{fact}
\label{fact:fiboeps}
For all~$n\geq 2$ we have
\[
 \frac{1}{2} \ \frac{\varphi^n}{\sqrt{5}}\ \leq \ f_n \ \leq \ \frac{3}{2} \ \frac{\varphi^n}{\sqrt{5}}.
\]
\end{fact}
\begin{proof}
Let us first show that we have
  \begin{equation}
    \label{eqn:encadre}
\frac{\varphi^n}{\sqrt{5}} - \frac{1}{2} \leq f_n \leq \frac{\varphi^n}{\sqrt{5}} + \frac{1}{2}.
  \end{equation}
  To show this, let 
  us consider the quantity $\frac{\log(\frac{\sqrt{5}}{2})}{\log(\varphi -
  1)}\simeq -0.2$, which is negative (note that $\varphi-1 > 0$ so the quantity
  is well-defined). As $n \geq 2$, we therefore have 
  $n \geq \frac{\log(\frac{\sqrt{5}}{2})}{\log(\varphi -
  1)}$. Multiplying by $\log (\varphi-1)$, which is negative, we have:
  \[n \log (\varphi-1) \leq \log(\frac{\sqrt{5}}{2}).\]
  The exponential is an increasing function, so we can exponentiate and get:
  \[(\varphi-1)^n \leq \frac{\sqrt{5}}{2}.\]
  The left-hand-side is clearly positive because $\varphi-1>0$, so:
  \[\left|(\varphi-1)^n\right| \leq \frac{\sqrt{5}}{2}.\]
  Thus:
  \[\frac{\left|(\varphi-1)^n\right|}{\sqrt{5}} \leq \frac{1}{2}.\]
  Binet's formula allows us to get Equation~\ref{eqn:encadre}.

Now $\frac{\varphi^n}{\sqrt{5}} \geq 1$ for $n\geq 2$, hence~$\frac{1}{2} \leq
\frac{1}{2} \, \frac{\varphi^n}{\sqrt{5}}$, and therefore
\[
 \frac{1}{2} \, \frac{\varphi^n}{\sqrt{5}} \leq f_n \leq \frac{3}{2} \, \frac{\varphi^n}{\sqrt{5}}
\]
when~$n\geq 2$, just as claimed.
\end{proof}

Indeed, these will allow us to show, for each inequality that we want to prove
to be true for all $n\geq m$, that it holds for all $n\geq n_\alpha$ for some
$n_\alpha$, and then to check by direct
computation that the inequality is also true for all~$n \in [m,n_\alpha]$. The
script where these computations are done can be found as
\verb|Fibonacci-inequalities.ipynb| in the supplementary material~\cite{supmat}.  Let us proceed.

\begin{lemma}
\label{lem:2nf2-f0f0}
We have $2^n  \geq \frac{f_n^2}{f_{n-2}}$ for all~$n \geq 4$.
\end{lemma}
\begin{proof}
By Fact~\ref{fact:fiboeps}, the following inequality is true for~$n \geq 4$:
\begin{align*}
  \frac{f_n^2}{f_{n-2}} &\leq \frac{\big(\frac{3}{2}\frac{\varphi^n}{\sqrt{5}}\big)^2}{\frac{1}{2}\frac{\varphi^{n-2}}{\sqrt{5}}}\\
 & = \frac{9\varphi^2 }{2\sqrt{5}} \times \varphi^n.\\
\end{align*}
We now use Fact~\ref{fact:2-vs-phi} with~$\alpha \colonequals \frac{9\varphi^2 }{2\sqrt{5}}$, and obtain~$n_\alpha \simeq 7.8$. Hence for~$n
\geq 8$ we have indeed $\frac{f_n^2}{f_{n-2}} < 2^n$. We prove that this is
also the case for~$n\in [4,7]$ by direct computation.
\end{proof}

\begin{lemma}
\label{lem:2nf2f0-f1f1f_1}
We have $2^n  \geq \frac{f_{n-1}^2 f_{n+1}}{f_{n-2} f_n}$ for all $n\geq 4$.
\end{lemma}
\begin{proof}
By Fact~\ref{fact:fiboeps}, the following inequality is true for~$n \geq 4$:
\begin{align*}
 \frac{f_{n-1}^2 f_{n+1}}{f_{n-2} f_n} &\leq
\frac{(\frac{3}{2} \ \frac{\varphi^{n-1}}{\sqrt{5}})^2 (\frac{3}{2} \ \frac{\varphi^{n+1}}{\sqrt{5}})}{(\frac{1}{2} \ \frac{\varphi^{n-2}}{\sqrt{5}})(\frac{1}{2} \ \frac{\varphi^{n}}{\sqrt{5}})}\\
& = \frac{27\varphi }{2\sqrt{5}} \times \varphi^n.
\end{align*}
We now use Fact~\ref{fact:2-vs-phi} with~$\alpha \colonequals \frac{27\varphi }{2\sqrt{5}}$, and obtain~$n_\alpha \simeq 10.7$. Hence for~$n
\geq 11$ we have indeed $\frac{f_{n-1}^2 f_{n+1}}{f_{n-2} f_n} < 2^n$. We prove
that this is also the case for~$n\in [4,10]$ by direct computation.
\end{proof}

\begin{lemma}
\label{lem:forsigma}
We have $2^n \geq 10 \times \frac{f_n^5}{f_{n-2}^4}$ for~$n \geq 48$.
\end{lemma}
\begin{proof}
By Fact~\ref{fact:fiboeps}, the following inequality is true for~$n \geq 4$:
\begin{align*}
10 \times \frac{f_n^5}{f_{n-2}^4} &\leq 10 \times \frac{(\frac{3}{2} \ \frac{\varphi^{n}}{\sqrt{5}})^5}{(\frac{1}{2} \ \frac{\varphi^{n-2}}{\sqrt{5}})^4}\\
 &= \frac{1215\ \varphi^8}{\sqrt{5}} \times \varphi^n.
\end{align*}
We now use Fact~\ref{fact:2-vs-phi} with~$\alpha \colonequals \frac{1215\
\varphi^8}{\sqrt{5}}$ and obtain~$n_\alpha \simeq 47.9 $.  Hence for~$n \geq 48$ we
have indeed $10 \times \frac{f_n^5}{f_{n-2}^4} \leq 2^n$.  
\end{proof}

\subsection{Step (i): An equivalent system}
\label{apx:fibo-equiv}
In this section we derive a system of equations that is equivalent
to~\ref{itemb}.

For brevity, we write $\Pi^{bb'}_{4}$ for $\Pi^{bb'}_{4}(p,q,r,s)$, and write
$\bar{p}$ for~$1-p$ (and similarly for~$q,r$ and~$s$).
Note that we can explicitly express $\Pi^{00}_{4}$ as a sum of products of the
$p,q,r,s$ and of the $\bar{p},\bar{q},\bar{r},\bar{s}$, intuitively
corresponding to the edge subsets of the graph $P_4(p,q,r,s)$ that are a
matching, and the same is true for $\Pi^{bb'}_{4}$ by removing the terms
involving $p$ (if $b = 1$) and those involving $s$ (if $b' = 1$).

We start by writing
out explicitly system~\ref{itemb}, which by Lemma~\ref{lem:fiboexpr}
consists of the following four equations:

\begin{description}
\item[(B)]:
\begin{align}
  \label{eqn:o1} \Pi^{00}_{4} & = \frac{f_{i+2}}{2^i}\\
  \label{eqn:o2} \Pi^{01}_{4} & = \frac{f_{i+1}}{2^i}\\
  \label{eqn:o3} \Pi^{10}_{4} & = \frac{f_{i+1}}{2^i}\\
  \label{eqn:o4} \Pi^{11}_{4} & = \frac{f_{i}}{2^i}
\end{align}
\end{description}

We claim that for any tuple~$(p,q,r,s)$ of real numbers,
this tuple satisfies~\ref{itemb} if and only if it satisfies the following system, denoted~\ref{iteme}.

\begin{description}
\item[(E)\label{iteme}]:
\begin{align}
  \label{eqn:g1} \pp \nq \ns & = \frac{f_{i-1}}{2^i}\\
  \label{eqn:g2} \ps \nr \np & = \frac{f_{i-1}}{2^i}\\
  \label{eqn:g3} \pq \pr & = \frac{1}{f_{i-1}^2}\\
  \label{eqn:g4} \np \ns & = \frac{f_{i-1}^2}{2^i f_{i-2}} 
\end{align}
\end{description}

We prove each implication in turn in the next two sections. Note that we
actually only need the direction \ref{iteme} $\implies$ \ref{itemb} for the
proof of Proposition~\ref{prp:fibo}, but we prove the equivalence for
completeness. 

\subsubsection{\ref{iteme} implies~\ref{itemb}}
We prove here that~\ref{iteme} implies~\ref{itemb}. 

\subparagraph*{Getting Equation \eqref{eqn:o4}.}
Observe that we have:
\begin{align}
\label{eqn:11}
 \Pi^{11}_{4} = \np \ns (1-qr).  
\end{align}
This is by explicit computation on $P_4(p,q,r,s)$: the edge subsets that are
a matching must have the first and last edge missing, and cannot have both
remaining edges present.

Now, we use Equation~\eqref{eqn:g4} and Equation~\eqref{eqn:g3} to substitute
$\np \ns$ and $qr$, and obtain:
\[
  \Pi^{11}_{4} = \frac{f_{i-1}^2}{2^i f_{i-2}} \left(1-
  \frac{1}{f_{i-1}^2}\right).  
\]
Let us bring to the same denominator. We get:
\[
  \Pi^{11}_{4} = \frac{f_{i-1}^2}{2^i f_{i-2}} \times
  \frac{f_{i-1}^2-1}{f_{i-1}^2}
\]
We use Cassini's identity, remembering that $i$ is even so $i-1$ is odd, and get:
\[
  \Pi^{11}_{4} = \frac{f_{i-1}^2}{2^i f_{i-2}} \times
  \frac{f_{i}f_{i-2}}{f_{i-1}^2}
\]
Simplifying, we get Equation~\eqref{eqn:o4}.

\subparagraph*{Getting Equations \eqref{eqn:o2} and \eqref{eqn:o3}.}
Next, observe that we have:
\begin{align}
  \label{eqn:01}\Pi^{01}_{4} = \Pi^{11}_{4}  + \pp \nq \ns.
\end{align}
Indeed, the edge subsets of $P_4$ that are a matching and have the last edge
missing are those that are a matching and have the first and last edge missing,
corresponding to $\Pi^{11}_{4}$, plus the ones where the last edge is missing
and the first edge is not missing. The latter implies that the second edge must
be missing, corresponding to the term $\pp \nq \ns$.

Now, we have just shown that 
$\Pi^{11}_{4} = \frac{f_{i}}{2^i}$, and we have $\pp \nq \ns =
\frac{f_{i-1}}{2^i}$ by Equation~\eqref{eqn:g1}. So we have:
\[
  \Pi^{01}_{4} = \frac{f_{i}}{2^i} + \frac{f_{i-1}}{2^i}
\]
By definition of the Fibonacci sequence ($f_{i+1} = f_i + f_{i-1}$), we obtain 
Equation~\eqref{eqn:o2}.
 
Equation~\eqref{eqn:o3} is
obtained symmetrically, using Equation~\eqref{eqn:g2} and the fact that:
\begin{align}
  \label{eqn:10}\Pi^{10}_{4} = \Pi^{11}_{4} + \ps \nr \np.
\end{align}

\subparagraph*{Getting Equation~\eqref{eqn:o1}.}
Last, we have:
\begin{align}
  \label{eqn:00}  \Pi^{00}_{4} = \Pi^{11}_{4} + \pp \ps \nq \nr + \pp \nq \ns + \ps \nr \np.
\end{align}
This is because the edge subsets of~$P_4$ that are matchings are those where
both the first and last edge are missing, corresponding to $\Pi^{11}_{4}$, plus
those where the first and last edge are present (corresponding to $\pp \ps \nq
\nr$ because the other edges must be missing), plus those where the first edge
is present and the last one missing (corresponding to $\pp \nq \ns$ because the
second edge must be missing), plus those where the last edge is present and the
first one is missing (corresponding analogously to $\ps \nr \np$).

By Equations~\eqref{eqn:o4}, \eqref{eqn:g1}, and \eqref{eqn:g2}, this means that
\begin{align}
  \label{eqn:bloub}\Pi^{00}_{4} = \frac{f_i}{2^i} + \pp \ps \nq \nr + \frac{f_{i-1}}{2^i} + \frac{f_{i-1}}{2^i}.
\end{align}
We now need to get rid of the term $\pp \ps \nq \nr$. But let us now multiply
Equations~\eqref{eqn:g1},~\eqref{eqn:g2} and~\eqref{eqn:g3}. We obtain
\[
\pp \np \pq \nq \pr \nr \ps \ns = \frac{1}{2^{2i}}.
\]
This implies that 
\[
\pp \ps \nq \nr = \frac{1}{2^{2i}(\np \ns) (\pq \pr)},
\]
and using Equations~\eqref{eqn:g3} and~\eqref{eqn:g4} we get 
\[
\pp \ps \nq \nr = \frac{f_{i-2}}{2^i}.
\]
We now inject the above into \eqref{eqn:bloub} and obtain
\begin{align*}
 \Pi^{00}_{4} &= \frac{f_i}{2^i} + \frac{f_{i-2}}{2^i} + \frac{f_{i-1}}{2^i} + \frac{f_{i-1}}{2^i}\\
	       &= \frac{f_{i+2}}{2^i},
\end{align*}
where the last line is obtained by applying three times the definition of the
Fibonacci sequence. We have thus obtained Equation~\eqref{eqn:o1}.

\subsubsection{\ref{itemb} implies~\ref{iteme}}

In this section we prove that~\ref{itemb} implies~\ref{iteme}.  We point out
again that this is not strictly necessary for the proof of
Proposition~\ref{prp:fibo} and can safely be skipped; we only include this for
completeness.

\subparagraph*{Getting Equations \eqref{eqn:g1} and \eqref{eqn:g2}.}
By Equations~\eqref{eqn:o2} and~\eqref{eqn:o3}, we have $\Pi^{01}_{4} = \Pi^{10}_{4}$. Hence by
Equations~\eqref{eqn:01} and~\eqref{eqn:10}:
\[
  \pp \nq \ns = \ps \nr \np.
\]
Calling this quantity $Q$, Equations~\eqref{eqn:01} and~\eqref{eqn:10} rewrite
to:
\[
  \Pi^{01}_{4} = \Pi^{10}_{4} = Q + \Pi^{11}_{4}.
\]
Now, by system~\ref{itemb} and by the definition
of the Fibonacci sequence, we have
\begin{align*}
\Pi^{01}_{4} &= \frac{f_{i+1}}{2^i} \\
	      &= \frac{f_i}{2^i} + \frac{f_{i-1}}{2^i}\\
	      &= \Pi^{11}_{4} + \frac{f_{i-1}}{2^i}.
\end{align*}
This implies that $Q = \pp \nq \ns = \ps \nr \np = \frac{f_{i-1}}{2^i}$, i.e.,
we have Equations~\eqref{eqn:g1} and~\eqref{eqn:g2}.

\subparagraph*{Getting Equation \eqref{eqn:g3}.}

Let us now compute the ``determinant-style'' expression $\Pi^{01}_{4}\Pi^{10}_4 - \Pi^{00}_{4} \Pi^{11}_{4}$. Combining Equations \eqref{eqn:01}, \eqref{eqn:10}, and \eqref{eqn:00}, we have:
\[
\Pi^{01}_{4}\Pi^{10}_4 - \Pi^{00}_{4} \Pi^{11}_{4} = (\Pi^{11}_{4}  + \pp \nq \ns)(\Pi^{11}_{4} + \ps \nr \np) - (\Pi^{11}_{4} + \pp \ps \nq \nr + \pp \nq \ns + \ps \nr \np)\Pi^{11}_{4}
\]
This simplifies, leaving us with:
\[
\Pi^{01}_{4}\Pi^{10}_4 - \Pi^{00}_{4} \Pi^{11}_{4} =  \pp \nq \ns \ps \nr \np - \pp \ps \nq \nr \Pi^{11}_{4}.
\]
We now inject \eqref{eqn:11} to obtain
\[
\Pi^{01}_{4}\Pi^{10}_4 - \Pi^{00}_{4} \Pi^{11}_{4} =
\pp \nq \ns \ps \nr \np
- \pp \ps \nq \nr \np \ns (1 - \pq \pr).
\]
This simplifies again, and we obtain, after reordering of terms:
\[
\Pi^{01}_{4}\Pi^{10}_4 - \Pi^{00}_{4} \Pi^{11}_{4} =
 \pp \pq \pr \ps \np \nq \nr \ns.
\]
On the other hand, by system~\ref{itemb} we have that
\[
2^{2i} (\Pi^{01}_{4}\Pi^{10}_4 - \Pi^{00}_{4} \Pi^{11}_{4}) = f_{i+1}^2 - f_{i+2} f_{i}.
\]
Now, remembering that~$i$ is even, we have by Cassini's identity that~$f_{i+1}^2 = f_i f_{i+2} + 1$.
Therefore, we obtain that
\begin{equation}
\label{eqn:f1}
\pp \pq \pr \ps \np \nq \nr \ns = \frac{1}{2^{2i}}.
\end{equation}
Continuing, notice that
\begin{align*}
\pp \pq \pr \ps \np \nq \nr \ns &= (\pp \nq \ns)(\ps \nr \np)\pq \pr\\
 &= \pq \pr \frac{f_{i-1}^2}{2^{2i}}
\end{align*}
by Equations \eqref{eqn:g1} and \eqref{eqn:g2}. Together with \eqref{eqn:f1}, we obtain $\pq \pr =
\frac{1}{f_{i-1}^2}$, that is, Equation~\eqref{eqn:g3}.

\subparagraph*{Getting Equation \eqref{eqn:g4}.}
We know that $\Pi^{11}_{4} = \np \ns (1-qr)$, and $\Pi^{11}_{4} =
\frac{f_i}{2^i}$ by Equation~\eqref{eqn:o4}. Hence, we have
\[
  \np \ns = \frac{f_i}{2^i (1-\pq \pr)}.
\]
We now use \eqref{eqn:g3} and Cassini's identity to simplify $1-qr$, and we
obtain $\np \ns = \frac{f_{i-1}^2}{2^i f_{i-2}}$, that is,
Equation~\eqref{eqn:g4}.  This concludes the proof that system~\ref{itemb} and
system~\ref{iteme} are equivalent.

\subsection{Step (ii): Membership in $(0,1)$}
\label{apx:fibo-bornes}

We now prove that \emph{any} tuple of real numbers $(p,q,r,s)$ that satisfies
system~\ref{iteme} must be
in~$(0,1)^4$. Clearly, looking at the equations of~\ref{iteme}, none of
$p,q,r,s$ can be equal to~$0$ or to~$1$.  Let us then consider three possible
ranges for each variable that cover all the possibilities: the variable is
either in~$(-\infty, 0)$, or it is in $(0,1)$, or it is in~$(1,\infty)$.
Observe that each equation of \ref{iteme} tells us something about the possible
ranges of~$p,q,r,s$, by considering that all expressions are positive:
for instance, Equation~\eqref{eqn:g3} implies that $q\in (-\infty,0)$ iff $r\in
(-\infty,0)$, and Equation~\eqref{eqn:g4} that~$p\in (1,\infty)$ iff~$s\in
(1,\infty)$. We use a helper script to analyze the $3^4 = 81$ possible ranges
of~$p,q,r,s$, shown in 
Algorithm~\ref{algo:prune-ranges}.
This algorithm as a Python script can be found as \verb|prune-ranges.py| in
supplementary material. 
\begin{algorithm}
\caption{Small helper script to prune out possible ranges of~$p,q,r$ and $s$.}
\KwIn{Nothing.}
\KwOut{Prints out the possible ranges that are not discarded by simple considerations on the sign of expressions in system~\ref{iteme}.}
\BlankLine
\tcc{We can use $-0.5$, $0.5$ and $1.5$ as representatives of the three
  candidate ranges for each variable}
\For{$(p,q,r,s)$ in $\{-0.5, 0.5, 1.5\}^4$}{
  \If(\tcp*[f]{By Equation~\eqref{eqn:g1}}){$p(1-q)(1-s) < 0$}{
    continue\;
  }
  \If(\tcp*[f]{By Equation~\eqref{eqn:g2}}){$s(1-r)(1-p) < 0$}{
    continue\;
  }
  \If(\tcp*[f]{By Equation~\eqref{eqn:g3}}){$qr < 0$}{
    continue\;
  }
  \If(\tcp*[f]{By Equation~\eqref{eqn:g4}}){$(1-p)(1-s) < 0$}{
    continue\;
  }
  print(($p$,$q$,$r$,$s$))\;
}
\tcc{This outputs 
$(-0.5, 1.5, 1.5, -0.5)$,
$(1.5, 1.5, 1.5, 1.5)$, 
$(0.5, -0.5, -0.5, 0.5)$,
$(0.5, 0.5, 1.5, -0.5)$,
$(-0.5, 1.5, 0.5, 0.5)$, and
$(0.5, 0.5, 0.5, 0.5)$.}
\label{algo:prune-ranges}
\end{algorithm}

Only $6$ ranges survive this script. The last one of them is~$(0,1)^4$, and we will show
that the five first ranges are in fact impossible,
by chasing down inequalities from~\ref{iteme} until we reach contradictions.

\subparagraph*{First two cases.} In the first two cases, we have~$q,r \in
(1,\infty)$.  But this is clearly not possible by Equation~\eqref{eqn:g3},
since~$\frac{1}{f_{i-1}^2} < 1$ (given that~$i\geq 4$).

\subparagraph*{Third case.} The third case to discard is when $p,s \in (0,1)$
and~$q,r \in (-\infty,0)$. Call this assumption~$(\dagger)$.  Observe then
that~$\nq >1$, and multiply this inequality by~$\pp$ and by $\ns$ (which are
both $>0$) to obtain
\[
 \nq \pp \ns > \pp \ns.
\]
Now, by \eqref{eqn:g1}, this implies:
\[
\frac{f_{i-1}}{2^i} > \pp \ns.
\]
Hence $\ns < \frac{1}{p} \times \frac{f_{i-1}}{2^i}$. Multiply by $\np$ (which is $>0$) and use \eqref{eqn:g4} to get
\[
 \frac{f_{i-1}^2}{2^i f_{i-2}} = \np \ns < \frac{\np}{\pp} \times \frac{f_{i-1}}{2^i}
\]
and so $\frac{\np}{\pp} > \frac{f_{i-1}}{f_{i-2}}$. 
As $p > 0$, we then have:
\[
\np > p \frac{f_{i-1}}{f_{i-2}}.
\]
Recalling that $p = 1 - \np$, we have:
\[
  \np > (1-\np) \frac{f_{i-1}}{f_{i-2}}.
\]
Rearranging terms, we obtain:
\[
\np > \frac{\frac{f_{i-1}}{f_{i-2}}}{1 + \frac{f_{i-1}}{f_{i-2}}}.
\]
We can simplify by multiplying the numerator and denominator by~$f_{i-2}$ and
using the definition of the Fibonacci sequence, to get:
\begin{align}
\label{ineqa}
\np > \frac{f_{i-1}}{f_{i}}.
\end{align}
This in turn implies, multiplying by~$\ns$ (which is~$>0$) and by
\eqref{eqn:g4} again, that
\[
\frac{f_{i-1}^2}{2^i f_{i-2}} = \np \ns > \ns \times \frac{f_{i-1}}{f_{i}}
\]
so that
\[
  \frac{f_{i} f_{i-1}}{2^i f_{i-2}} > \ns.
\]
Now, using $\ns = 1 - \ps$ and reordering terms, we get that
\[
 \ps > \frac{2^i f_{i-2} - f_i f_{i-1}}{2^i f_{i-2}}.
\]
But $2^i f_{i-2} > f_i f_{i-1}$ by Lemma~\ref{lem:2nf2-f0f0}
(using~$f_i > f_{i-1}$), and~$\ps > 0$ by $(\dagger)$, hence we have
\begin{align}
 \label{ineqb}\frac{1}{\ps} < \frac{2^i f_{i-2}}{2^i f_{i-2} - f_i f_{i-1}}.
\end{align}
Now, we multiply \eqref{ineqa} by $\ps$ and $\nr$ (which are $>0$) to obtain
\[
  \frac{f_{i-1}}{2^i} = \np \ps \nr > \ps \nr \times \frac{f_{i-1}}{f_i},
\]
hence $\nr < \frac{1}{s} \times \frac{f_i}{2^i}$, and thanks to \eqref{ineqb} we get
\[
\nr < \frac{f_i f_{i-2}}{2^i f_{i-2} - f_i f_{i-1}},
\]
implying 
\[
 \pr > \frac{2^i f_{i-2} - f_i f_{i-1} - f_i f_{i-2}}{2^i f_{i-2} - f_i f_{i-1}}
 = \frac{2^i f_{i-2} - f_i^2}{2^i f_{i-2} - f_i f_{i-1}}
\]
where the last equality is using the definition of the Fibonacci sequence.
This last expression is positive according to
Lemma~\ref{lem:2nf2-f0f0}. But~$\pr$ is supposed to be negative by
$(\dagger)$, a contradiction.

\subparagraph*{Fourth case.} The fourth case to discard is when $s \in
(-\infty,0)$, $p,q \in (0,1)$ and $r \in (1,\infty)$.  Call again this
assumption $(\dagger)$. 
We start by multiplying $\pr >1$ by
$q$, which is positive by~$(\dagger)$, to get, with \eqref{eqn:g3}, 
\[
 \frac{1}{f_{i-1}^2} = \pq \pr > \pq,
\]
hence
\[
 \frac{1}{f_{i-1}^2} > 1 - \nq,
\]
so that reordering terms
\[
  \nq > 1 - \frac{1}{f_{i-1}^2} = \frac{f_{i-1}^2 - 1}{f_{i-1}^2},
\]
and by Cassini's identity we get:
\[\nq > \frac{f_i f_{i-2}}{f_{i-1}^2}.
\]
We multiply this last
inequality by $\pp$ and $\ns$, which are positive by~$(\dagger)$, and obtain,
from \eqref{eqn:g1},
\[
 \frac{f_{i-1}}{2^i} = \nq \pp \ns > \pp \ns \times \frac{f_i f_{i-2}}{f_{i-1}^2}.
\]
Therefore $\ns < \frac{1}{p} \times \frac{f_{i-1}^3}{2^i f_{i-2} f_i}$. Multiply by
$\np$, which is positive by~$(\dagger)$, and use \eqref{eqn:g4} to obtain
\[
 \frac{f_{i-1}^2}{2^i f_{i-2}} = \np \ns < \frac{\np}{\pp} \times \frac{f_{i-1}^3}{2^i f_{i-2} f_i},
\]
which simplifies to $\frac{\np}{\pp} > \frac{f_i}{f_{i-1}}$. 
As $p > 0$, and by similar reasoning as in the previous case, we get that
$\np >  \frac{f_i}{f_{i+1}}$.
Multiply by $\ns$ and use \eqref{eqn:g4} again to obtain
\[
 \frac{f_{i-1}^2}{2^i f_{i-2}} = \np \ns > \ns \times \frac{f_i}{f_{i+1}},
\]
hence
\[
  \ns < \frac{f_{i-1}^2 f_{i+1}}{2^i f_{i-2} f_i}
\]
and, using $\ns = 1 - \ps$ and reordering, we get:
\[
 \ps > \frac{2^i f_{i-2} f_i - f_{i-1}^2 f_{i+1}}{2^i f_{i-2} f_i}.
\]
But this last expression is positive by
Lemma~\ref{lem:2nf2f0-f1f1f_1}, whereas $s$ is negative according to
$(\dagger)$, a contradiction.

\subparagraph*{Fifth case.} The last case to discard is when $p \in
(-\infty,0)$, $q \in (1,\infty)$ and $r,s \in (0,1)$. But this case is
symmetrical to the fourth case,
so we are done.

\subsection{Step (iii): Satisfying the System}
\label{apx:fibo-sat}

We use SageMath to obtain a solution. The code can be found in the Jupyter
notebook \verb|obtain-solution-and-check-sigma-small-N.ipynb| in supplementary material~\cite{supmat}.  Writing~$T
\colonequals 1/2^i$ and~$F_k \colonequals f_{i+k}$, we define:

\begin{align*}
P &\colonequals 2 \, F_{-1} F_{-2}^{2} + 2 \, {\left(F_{-1}^{2} - 1\right)} F_{-2} \\
Q &\colonequals 2 \, F_{-1}^{2} F_{-2} - 2 \, {\left(F_{-1}^{4} + F_{-1}^{3} F_{-2}\right)} T \\
A &\colonequals 2 \, F_{-1} F_{-2}^{2} \\
\Xi &\colonequals F_{-1}^{2} F_{-2} - {\left(F_{-1}^{4} + 2 \, F_{-1}^{3} F_{-2} + F_{-1}^{2} F_{-2}^{2}\right)} T \\
\Theta &\colonequals F_{-1}^{2} T - F_{-2} \\
 C_0 &\colonequals \left(F_{-1}^{4} - 2 \, F_{-1}^{2} + 1\right) F_{-2}^{2}\\
 C_1 &\colonequals 2 \, \left({\left(F_{-1}^{4}\, +\, F_{-1}^{2}\right)} F_{-2}^{3}\, +\, 2 \, {\left(F_{-1}^{5} - F_{-1}^{3}\right)} F_{-2}^{2}\, +\, {\left(F_{-1}^{6} - 2 \, F_{-1}^{4}\, +\, F_{-1}^{2}\right)} F_{-2}\right) \\
 C_2 &\colonequals  F_{-1}^{8} + 4 F_{-1}^{5} F_{-2}^{3} + F_{-1}^{4} F_{-2}^{4} - 2 F_{-1}^{6} + F_{-1}^{4}\\
     &\phantom{=} + 2 (3 F_{-1}^{6} - F_{-1}^{4}) F_{-2}^{2} + 4 (F_{-1}^{7} - F_{-1}^{5}) F_{-2}\\
 \Sigma &\colonequals C_0 - C_1 T + C_2 T^2.
\end{align*}

\begin{remark}
These expressions could be simplified, using the properties of the Fibonacci
sequence. Nevertheless, we leave them as-is in this section so that it is
easier to see that they match the ones in the notebook, and to keep the
notebook as clean as possible.  We will simplify some of these in the next
section.  
\end{remark}

Finally, we pose
\begin{align*}
p(i) &\colonequals (A + \Xi + \Theta + \sqrt{\Sigma})/P\\
q(i) &\colonequals (\Xi - \Theta + \sqrt{\Sigma})/Q\\
r(i) &\colonequals (\Xi - \Theta - \sqrt{\Sigma})/Q\\
s(i) &\colonequals (A + \Xi + \Theta - \sqrt{\Sigma})/P
\end{align*}

One can then check by executing the notebook
that these expressions satisfy system \ref{iteme}, when we look at~$F_{-1}$ and
$F_{-2}$ as symbolic variables. Note that one can understand how this
verification can be performed; when computing the left-hand sides we obtain fractions whose
numerator is a polynomial $P_1$ in $T$ and the $F_k$, plus such a polynomial
$P_2$ times a square root of such a polynomial, divided by such a polynomial. Up
to multiplying by the denominator, showing the identities amounts to checking
that the polynomials $P_1$ and $P_2$ are correct, which can be done by expanding
them and checking that the monomials are correct.

We point, however, that the resulting expressions are only
\emph{symbolic} expressions, and that we do not know a priori if they are
well-defined, that is, if $\Sigma$ is always non-negative and if the denominators
are not null. We prove in the next section that this is the case, using
properties of the Fibonacci sequence.

\subsection{Step (iv): Checking Well-Definedness}
\label{apx:fibo-welldef}

We now prove that the expressions $p(i),q(i),r(i),s(i)$ are well-defined, i.e.,
that~$\Sigma$ is non-negative and that the denominators~$P$ and~$Q$ are never
null. This will effectively conclude the proof of Proposition~\ref{prp:fibo}.
Remember that~$i$ is an even integer greater than~$4$. 

\subparagraph*{Checking that $P$ is not null.} By Cassini's identity we have~$F_{-1}^2 =
F_{-2} F_0 + 1$, and so we get 
\begin{align*}
 P &= 2 \, F_{-1} F_{-2}^{2} + 2 \,  F_{-2}^2 F_0 \\
   &= 2 \, F_{-2}^2 (F_{-1} + F_0)\\
   &= 2\, F_{-2}^2 F_1,
\end{align*}
hence~$P$ is clearly not null.

\subparagraph*{Checking that $Q$ is not null.} 
We have:
\begin{align*}
& Q = 0\\
& \iff 2\, F_{-1}^2 F_{-2} = 2 \, T \, (F_{-1}^4 + F_{-1}^3 F_{-2})\\
& \iff F_{-2} = T \, F_{-1} \, (F_{-1} + F_{-2})\\
& \iff 2^i = \frac{F_{-1} F_0}{F_{-2}}
\end{align*}
But $2^i > \frac{F_{-1} F_0}{F_{-2}}$ according to
Lemma~\ref{lem:2nf2-f0f0} (using~$F_0 > F_{-1}$), so~$Q$ cannot be null. 

\subparagraph*{Checking that $\Sigma$ is non-negative.} Recall that~$\Sigma = C_0 -C_1 T +
C_2 T^2$, where~$T$ does not occur in any of the~$C_i$.  First, we
simplify~$C_0$, using Cassini's equality.
\begin{align*}
 C_0 &= \left(F_{-1}^{4} - 2 \, F_{-1}^{2} + 1\right) F_{-2}^{2}\\
     &= F_{-2}^2 \, (F_{-1}^2 - 1)^2\\
     &= F_{-2}^2 \, (F_{-2} F_0)^2\\
     &= F_{-2}^4  F_0^2.
\end{align*}

Next, we crudely upper bound~$C_1$, again using Cassini's identity, and the monotonicity
of the Fibonacci sequence. 
\begin{align*}
 C_1 &= 2 \, \left({\left(F_{-1}^{4}\, +\, F_{-1}^{2}\right)} F_{-2}^{3}\, +\, 2 \, {\left(F_{-1}^{5} - F_{-1}^{3}\right)} F_{-2}^{2}\, +\, {\left(F_{-1}^{6} - 2 \, F_{-1}^{4}\, +\, F_{-1}^{2}\right)} F_{-2}\right) \\
    &= 2 \left((F_{-1}^{4}\, +\, F_{-1}^{2}) F_{-2}^{3}\, +\, 2 \, F_{-2}^2 F_{-1}^3 (F_{-1}^2 -1)\, +\, F_{-2} F_{-1}^2 (F_{-1}^2 - 1)^2 \right)\\
    &= 2 \left((F_{-1}^{4}\, +\, F_{-1}^{2}) F_{-2}^{3}\, +\, 2 \, F_{-2}^3 F_{-1}^3 F_0\, +\, F_{-2}^3 F_{-1}^2 F_0^2 \right)\\
   &\leq 2 \left((F_{0}^{4}\, +\, F_{0}^{2}) F_{0}^{3}\, +\, 2 \, F_{0}^3 F_{0}^3 F_0\, +\, F_{0}^3 F_{0}^2 F_0^2 \right)\\
   &= 2 \, (F_0^5 + 4 F_0^7)\\
   &\leq 2 \, (F_0^7 + 4 F_0^7)\\
   &= 10 F_0^7.
\end{align*}

We then show that~$C_2 \geq0$ as follows: 
\begin{align*}
 C_2 &=  F_{-1}^{8} + 4 F_{-1}^{5} F_{-2}^{3} + F_{-1}^{4} F_{-2}^{4} - 2 F_{-1}^{6} + F_{-1}^{4} \\
     &\phantom{=} + 2 (3 F_{-1}^{6} - F_{-1}^{4}) F_{-2}^{2} + 4 (F_{-1}^{7} - F_{-1}^{5}) F_{-2}\\
     & = F_{-1}^6 (F_{-1}^2 - 2) + 4 F_{-1}^{5} F_{-2}^{3} + F_{-1}^{4} F_{-2}^{4}  + F_{-1}^{4} \\
     &\phantom{=} + 2 F_{-2}^{2} F_{-1}^4 (3 F_{-1}^{2} - 1) + 4 F_{-2} F_{-1}^{5} (F_{-1}^{2} - 1)\\
\end{align*}
But clearly~$F_{-1}^2 \geq2$ (since~$i \geq 4$ and~$F_3 = 2$), hence~$C_2 \geq0$ indeed.

But then, observe that this implies 
\[
\Sigma \geq F_{-2}^4  F_0^2 - 10 F_0^7 \times T.
\]
This last term is non-negative if and only if $2^i \geq 10 \times
\frac{F_0^5}{F_{-2}^4}$, and we prove this to be the case in
Lemma~\ref{lem:forsigma} for~$i \geq 48$. We complete the proof by checking by
direct computation that~$\Sigma \geq 0$ for all even integers between~$4$ and~$47$ 
(\verb|obtain-solution-and-check-sigma-small-N.ipynb| in supplementary material~\cite{supmat}).

\section{Proofs for the Precision Argument}
\label{apx:precision}
In this section we give the proof details for step 5 of Section~\ref{sec:general}.
Recall that our goal is to
determine the result $\mC$ of all oracle calls on the graphs
$G(\kappa,\kappa')$, because then we could recover~$\mS$ as $\mG^{-1} \mC$ since we can compute the
inverse of~$\mG$ in polynomial time.
The problem is that we cannot actually invoke the oracle on the graphs
$G(\kappa,\kappa')$, because some of the edge probabilities are
non-rational, namely, the $p(i),q(i),r(i),s(i)$.

To work around this issue, the first step is to argue that we can compute
decimal fraction approximations of the $p(i),q(i),r(i),s(i)$, in polynomial time in the
number of desired decimal places. This easily follows from their expressions of
the form $\frac{P\pm \sqrt{Q}}{R}$ with $P,Q,R$ polynomials in quantities that
we can compute exactly.

\begin{restatable}{lemma}{precision}
\label{lem:precision}
  Given an even integer $i \geq 4$ and number $z$ of decimal places, both written
  in unary, we can compute in polynomial time in~$i$ and $z$ four decimal fractions 
  $0 \leq \widehat{\pp(i)}, \widehat{\pq(i)}, \widehat{\pr(i)}, \widehat{\ps(i)} \leq 1$ such that $|\pp(i)-\widehat{\pp(i)}|
  \leq 2^{-z}$ and similarly for $\widehat{\pq(i)}$, $\widehat{\pr(i)}$, and $\widehat{\ps(i)}$.
\end{restatable}
\begin{proof}
First note that we can compute exactly the Fibonacci numbers, the value $2^{-i}$,
and polynomials in these values, as exact rationals. Now,
  we recall from Proposition~\ref{prp:fibo} that the expressions for $p(i)$, $q(i)$,
$r(i)$, and $s(i)$ are a sum of such polynomials and of $\sqrt{\Sigma}$, where
$\Sigma$ is such a polynomial, divided by such polynomials ($P$ or $Q$).
So the problem boils down to approximating expressions of the form $\frac{X \pm
\sqrt{\Sigma}}{Y}$ where $X$, $Y$, and $\Sigma$ are computed as exact rationals.

We know that there is some value~$l$ polynomial
in~$i$ such that 
 $|Y| > 2^{-l}$, because~$Y$ is computed in polynomial time, i.e., the number
 of decimal places of~$Y$ must be polynomial in~$i$. Let us approximate
 $\sqrt{\Sigma}$ to have error at most $2^{-(l+z)}$, which we can do in
 polynomial time in~$l$ and~$z$ (see, e.g., \cite{sqrtprecision}).
 The absolute error in the result is then at most $2^{-l}$, because we know
 $X/Y$ exactly and we know $\sqrt{\Sigma}/Y$ up to $2^{-z}$ since the error in
  $\sqrt{\Sigma}$ is at most multiplied by~$2^{l}$.
\end{proof}

We now point out that the oracle result $\mC$ that we wish to obtain is in fact
a vector of decimal fractions, and that we can bound its number of decimal places.
This is easy to notice if we consider the graph $\sub(H,\eta')$ where $\eta'$
subdivides each edge $e$ to $N$ or $N'$ depending on the parity of $\eta(e)$.
Indeed, proposition~\ref{prp:fibo} then ensures
that the oracle result $\mC$ on the “ideal” graphs $G(\kappa,\kappa')$
is the same result that we would obtain on the graph $\sub(H,\eta')$ with
probabilities set as in the lower part of Figure~\ref{fig:emul}, i.e., \emph{not} using
the $p(i),q(i),r(i),s(i)$. Now, as all the probabilities on that graph are
decimal fractions, the answer $\mC$ is in fact itself a vector of decimal fractions. Further, we can bound the number of
decimal places of its components to $[m \times (\max(N,N') + 10)] \times z$ with $z$ the maximal
number
of decimal places of a decimal fraction in $\rho_\kappa, \rho'_{\kappa}$, which
is polynomial. This uses the following immediate result:

\begin{restatable}{lemma}{precgraph}
  \label{lem:precgraph}
  Let $(G, \pi)$ be a probabilistic graph with $\pi$ mapping each edge $e\in E$
  to a probability $\pi(e)$ which is a decimal fraction with at most $z$
  decimal places. Then the number of decimal places of
  $\prm(G,\pi)$ is at most $m z$, where $m$ is the number of edges of~$m$.
\end{restatable}
\begin{proof}
  The answer to $\prm(G,\pi)$ is a sum over edge subsets of $(G, \pi)$, so it
  suffices to show the result for each edge subset. Now, the probability of an
  edge subset is a product of $m$ values which are either edge probabilities
  $\pi(e)$ or their complement $1 - \pi(e)$. Both of these have at most $z$
  decimal places, so the product has at most $mz$ decimal places, which concludes.
\end{proof}

Now, the last step is to argue that we can recover exactly the oracle result
$\mC$ by invoking the oracles on graphs with approximations of the non-rational
probabilities.
Let $z' \colonequals m
\times (\max(N,N') + 10) \times  z$ be the maximal number of decimal places of a
component of~$\mC$.
Let $z'' \colonequals z'+2m + 1$.
Use Lemma~\ref{lem:precision} to compute, for each $0 \leq i \leq
\max(N,N')$, decimal fraction approximations
$\widehat{\pp(i)},\widehat{\pq(i)},\widehat{\pr(i)},\widehat{\ps(i)}$ of
$\pp(i),\pq(i),\pr(i),\ps(i)$ which are accurate to $z''$ binary places.  Let
$\widehat G(\kappa,\kappa')$ be the graph defined like $G(\kappa,\kappa')$ but
with the probabilities $\pp(i),\pq(i),\pr(i),\ps(i)$ replaced by 
$\widehat{\pp(i)},\widehat{\pq(i)},\widehat{\pr(i)},\widehat{\ps(i)}$. We call
the oracle on these graphs, and obtain a vector $\widehat \mC$. The only
missing ingredient is to bound the error on each component of~$\widehat \mC$
relative to $\mC$. This follows from an easy variant of
Lemma~\ref{lem:precgraph}:

\begin{restatable}{lemma}{precgraphb}
  \label{lem:precgraphb}
  Let $(G, \pi)$ with $G = (V, E)$ be a probabilistic graph with $\pi$ mapping each edge $e\in E$
  to a real value, and let $\widehat \pi$ mapping each edge $e$ to a 
  decimal fraction $\widehat \pi(e)$. Then we have: 
  $|\prm(G,\pi) - \prm(G,\widehat\pi)| \leq 2^{2m} \max_{e\in E} |\widehat \pi(e) - \pi(e)|$.
\end{restatable}
\begin{proof}
  Let $W$ be the set of all edge subsets of $G$ which are matchings.
  Write $\delta(e) \colonequals \widehat\pi(e) - \pi(e)$ for $e \in E$.
  We have:
  \[
    \prm(G,\widehat\pi) - \prm(G,\pi) = \sum_{E' \in W} \prod_{e \in E'} \widehat\pi(e)
    \prod_{e \notin E'} (1 - \widehat\pi(e)) - \prm(G,\pi) 
  \]
  Thus, injecting $\delta(e)$:
  \[
    \prm(G,\widehat\pi) - \prm(G,\pi)= \sum_{E' \in W} \prod_{e \in E'} (\pi(e) +
    \delta(e)) \prod_{e \notin E'} ([1 - \pi(e)] - \delta(e)) - \prm(G,\pi)
  \]
  Let us bound the absolute value of the left-hand side to conclude. Note that,
  in the right-hand side, we can expand each term of the sum to obtain $2^m$
  terms, one of which is the corresponding term of $\prm(G,\pi)$, the others all
  having some $\delta(e)$ as a factor, and the other quantities in the product
  are between 
$-1$ and $1$ because this is the case of the values of $\pi(e)$, of
  $\widehat\pi(e)$, and of their differences $\delta(e)$. Hence, using the triangle
  inequality we obtain
  \[
    \left|\prm(G,\widehat\pi) - \prm(G,\pi)\right| \leq \sum_{E' \in W} 
    2^m \max_{e\in E} |\widehat \pi(e) - \pi(e)|,
  \]
  and the result follows immediately.
\end{proof}

Thus, the result $\widehat\mC$ of the oracle calls is such that each
component has error at most $2^{-(z'+1)}$ relative to $\mC$.
Thus, truncating them to $z'$ decimal places matches the exact value of $\mC$.
Thus, we can recover the values $\mC$
and conclude the proof.

\section{Proof of Theorem~\ref{thm:edgecovers}}
\label{apx:others}
In this section we give more details on the proof of
Proposition~\ref{thm:edgecovers} by explaining how the proof of
Sections~\ref{sec:sub6}--\ref{sec:general} and of the 
relevant lemmas and propositions is modified.

Using the new definition of the~$\Pi^{bb'}_n$ as given in
Section~\ref{sec:others}, Lemma~\ref{lem:fiboexpr} becomes:

\begin{lemma}
\label{lem:fiboedgecovers}
For all $n \in \mathbb{N}^+$, $b, b' \in \{0, 1\}$, we have
$\Pi^{bb'}_{n}(1/2,\ldots,1/2) = \frac{f_{n+b+b'}}{2^n}$.
\end{lemma}
\begin{proof}
We first prove it for~$b=b'=1$, by induction on~$n$. 
For~$n=1$, we have~$\Pi^{11}_1(1/2) = 1$, and~$f_{3}/2^1 = 1$. 
For~$n=2$ we have~$\Pi^{11}_2(1/2,1/2) = 3/4$, and~$f_{4}/2^2 = 3/4$. 
Now, assuming the claim holds for~$1\leq i \leq n-1$, $n\geq 3$, let us look
at~$\Pi^{11}_n(1/2,\ldots,1/2)$.  By case analysis on the first edge, we have:
\begin{align*}
\Pi^{11}_n(1/2,\ldots,1/2) &= \frac{1}{2} \Pi^{11}_{n-1}(1/2,\ldots,1/2) + \frac{1}{2} \frac{1}{2} \Pi^{11}_{n-2}(1/2,\ldots,1/2),
\end{align*}
and the induction hypothesis together with the definition of the Fibonacci sequence concludes.

Now, it is clear that~$\Pi^{01}_n(1/2,\ldots,1/2)=\frac{1}{2} \Pi^{11}_{n-1}(1/2,\ldots,1/2)$, because
the first edge must be present, and then the condition on the~$n-1$ remaining edges is the same as the condition of~$P_{n-1}(1/2,\ldots,1/2)$.
Similarly we have $\Pi^{10}_n(1/2,\ldots,1/2)=\frac{1}{2} \Pi^{11}_{n-1}(1/2,\ldots,1/2)$
and $\Pi^{00}_n(1/2,\ldots,1/2)=\frac{1}{2^2} \Pi^{11}_{n-2}(1/2,\ldots,1/2)$, so the claim holds.
\end{proof}

Lemma~\ref{lem:concat} becomes:

\begin{lemma}
\label{lem:concat2}
Let~$n,n' \in \mathbb{N}^+$ and~$\rho \in [0,1]^n$, $\rho' \in [0,1]^{n'}$ be tuples of probability values. 
Then, for every~$b,b' \in \{0,1\}$ we have
\begin{align*}
\Pi_{n+n'}^{bb'}(\rho, \rho') \ = \ & (\Pi_n^{b0}(\rho) \times \Pi_{n'}^{1b'}(\rho'))
+ \ (\Pi_n^{b1}(\rho) \times \Pi_{n'}^{0b'}(\rho'))
- \ (\Pi_n^{b0}(\rho) \times \Pi_{n'}^{0b'}(\rho')).
\end{align*}
\end{lemma}
\begin{proof}[Proof sketch]
  We only sketch the proof, as a formal proof similar to that of
  Lemma~\ref{lem:concat} is easy to obtain from this.
  The possible edge subsets in~$P_{n+n'}$ that are edge covers are those where
the $n$-th edge is kept, plus those where the $(n+1)$-th edge is kept,
minus those where both the $n$-th and $(n+1)$-th edge were kept (as these were
counted twice).
\end{proof}

Next, we argue that the analogue of Proposition~\ref{prp:fibo} still holds.
To this end, let us inspect the system of
Equation~\ref{itemb} from this proposition, where the~$\Pi^{bb'}_n$ are now defined
using edge covers. We first look at $\Pi^{00}_4(p,q,r,s)$. By explicit computation we have
\begin{align*}
 \Pi^{00}_{4}(p,q,r,s) = \pp \ps (1-\nq \nr).  
\end{align*}

Indeed, the edge subsets of~$P_4$ that are an edge cover must have the first
and last edge present, and cannot have both middle edges absent.
Notice then that this is the same as the “old” $\Pi^{11}_{4}(p,q,r,s)$ (for
matchings), but switching~$p,q,r,s$ into~$\np,\nq,\nr,\ns$ (see
Equation~\ref{eqn:11}).

We now look at $\Pi^{01}_{4}(p,q,r,s)$. We have
\begin{align}
  \Pi^{01}_{4}(p,q,r,s) = \Pi^{00}_{4}(p,q,r,s)  + \pp \pr \ns.
\end{align}
Indeed, the edge subsets of $P_4$ that are an edge cover and have the first edge
present are those that are an edge cover and have the first and last edge present,
corresponding to $\Pi^{00}_{4}(p,q,r,s)$, plus the ones where the first edge is present
and the last edge is missing. The latter implies that the third edge must
be present, corresponding to the term $\pp \pr \ns$.
Notice then that this is the same as the old $\Pi^{10}_{4}(p,q,r,s)$, but switching~$p,q,r,s$ into~$\np,\nq,\nr,\ns$ (see
Equation~\ref{eqn:10}).

Similarly, the new~$\Pi^{10}_{4}(p,q,r,s)$ is equal to the old $\Pi^{01}_{4}(p,q,r,s)$, but switching~$p,q,r,s$ into~$\np,\nq,\nr,\ns$ (Equation~\ref{eqn:01}).

In the same manner, we easily can show that the new~$\Pi^{11}_{4}(p,q,r,s)$ is equal
to the old $\Pi^{00}_{4}(p,q,r,s)$, but switching~$p,q,r,s$
into~$\np,\nq,\nr,\ns$ (Equation~\ref{eqn:00}).

Therefore, using Lemma~\ref{lem:fiboedgecovers}, the new System~\ref{itemb} can
be obtained from the old one by simply switching $p,q,r,s$
into~$\np,\nq,\nr,\ns$, hence we can take the solutions to the new system to
be~$(1-p(i),1-q(i),1-r(i),1-s(i))$, where $p(i),q(i),r(i),s(i)$ is the solution
to the old system: indeed, this still satisfies condition~\ref{itema}, and the
obtained solutions are also of the prescribed form.

Last, we point out that the relevant Jacobian determinants are again not the
null polynomials: this follows directly from the previous remarks on how the
$\Pi^{bb'}_{4}(p,q,r,s)$ have been changed, and from the fact that the
determinant is alternating. For instance, the new $\det(\mJ)$ is then the “old”
$\det(\mJ)$, but again swapping~$\chi_{00}, \chi_{01},\chi_{10},\chi_{11}$ by
$1-\chi_{00}, 1-\chi_{01},1-\chi_{10},1-\chi_{11}$: indeed, when we derive the
new $\Pi^{b_1 b_2}_{4}(\chi)$ by~$\chi_{b'_1,b'_2}$, only the sign changes.
Since it changes in every cell and the matrix is~$4\times4$, it does not change
globally.
An inspection of the rest of the proof reveals that it still works.

\end{document}